\documentclass[sn-mathphys-num]{sn-jnl}

\usepackage{xcolor}
\definecolor{darkblue}{rgb}{0.,0.,0.75}
\definecolor{darkred}{rgb}{0.5,0.,0.}
\definecolor{darkgreen}{rgb}{0.0,0.5,0.}
\usepackage{amsmath, amssymb,mathtools,amsthm,amsfonts,amssymb,braket,comment,mathtools,bbm, mathrsfs}
\usepackage[all]{xy}
\usepackage{fullpage}
\usepackage{chemarrow}
\usepackage{marginnote}
\usepackage[numbers]{natbib}
\usepackage{tikz-cd}
\usepackage{scalerel}
\newtheorem{lemma}{Lemma}[section]
\newtheorem{theorem}[lemma]{Theorem}

\newtheorem{corollary}[lemma]{Corollary}

\newtheorem{proposition}[lemma]{Proposition}
\newtheorem{proposition-definition}[lemma]{Proposition-Definition}
\theoremstyle{definition}
\newtheorem{definition}[lemma]{Definition}
\newtheorem{remark}[lemma]{Remark}
\newtheorem{example}[lemma]{Example}
\usepackage{enumitem}
\usepackage[normalem]{ulem}

\usepackage{appendix}
\numberwithin{equation}{section}
\theoremstyle{definition}
\newcommand{\norm}[1]{\left\|#1\right\|}
\makeatletter %% Spacing before and after Thm environment
\def\thm@space@setup{%
  \thm@preskip=6pt
  \thm@postskip=6pt
}
\makeatother

\usepackage{etoolbox}

\makeatletter
\patchcmd{\@maketitle}
  {{\ifnum\aucount>1*\fi}%
   Corresponding author(s). E-mail(s): \corrauthemail\par}
  {\corrauthemail\par}
  {}{\typeout{Warning: email-label patch failed}}
\makeatother

\def\CA{{\mathcal A}}
\def\CB{{\mathcal B}}

\def\CD {{\mathcal D}}

\def\CH {{\mathcal H}}

\def\CB {{\mathcal B}}
\def\CD {{\mathcal D}}

\def\CT {{\mathcal T}}

\def\IN{\mathbb{N}}

\def\IC{\mathbb{C}}

\def\IR{{\mathbb{R}}}
\def\IT{{\mathbb{T}}}
\def\IZ{{\mathbb{Z}}}

\def\Tr{{\operatorname{Tr}}}

\def\bZ{\mathbb{Z}}

\newcommand{\supp}{\operatorname{supp}}
\newcommand{\abs}[1]{\left\lvert #1\right\rvert}

\def\a{\alpha}

\def\La{{\Lambda}}
\def\l{{\lambda}}

\def\sc{\scaleto{\,\circ\,}{4pt}}
\def\sl{\mathcal{S}_{\textsuperscript{L}}}

\def\Det{\mbox{Det}}

\newcommand{\id}{\mathbbm{1}}

\begin{document}

\title[Thermalization with Gaussian QCAs]{Thermalization with Gaussian Quantum Cellular Automata}
%thermalization of QCA on continuous variables

\author{\fnm{Roman} \sur{Geiko}}\email{geiko@physics.ucla.edu}

\author{\fnm{Jake} \sur{Gerenraich}}\email{gerenraich@physics.ucla.edu}

\affil{\orgdiv{Department of Physics and Astronomy}, \orgname{University of California}, \orgaddress{\city{Los Angeles}, \postcode{90095}, \state{CA}, \country{USA}}}

\abstract{We study the long-time dynamics of many-body bosonic lattice systems under
translation-invariant Gaussian quantum cellular automata. We formulate two sets of conditions on GQCAs which separately guarantee  
thermalization of any state on the local Weyl algebra to the infinite temperature state, whenever the state is locally normal and has uniformly bounded particle density. Our main intermediate result is a quantum many-body generalization of the classic Riemann-Lebesgue lemma which is a bound on expectation values of local Weyl operators involving their support and the state's particle density.}


\maketitle

\section{Introduction and main results}

The emergence of thermodynamic behavior from unitary evolution of quantum systems is central to justifying the formalism of equilibrium statistical mechanics from first principles. Despite decades of progress in the foundations of quantum statistical mechanics, a robust understanding of thermalization, or ``approach to equilibrium'', has remained elusive. One path forward is to identify broad, physically motivated classes of dynamics for which thermalization can be rigorously established. Here we examine quantum lattice systems of infinite-dimensional degrees of freedom in infinite volume and formulate two sets of sufficient conditions on the dynamics to thermalize a certain class of initial states.   

We consider discrete time dynamics which can be seen as the Floquet, or stroboscopic, picture of a continuous time evolution. In particular, we suppose the dynamics have a uniform bound on the spread of the support of local operators, known as  Quantum Cellular Automata (QCAs) \cite{schumacher2004reversiblequantumcellularautomata}. Thermalization with QCAs was studied in e.g. \cite{Kruger_GQCA,G_tschow_2010,kapustin2026chaosthermalizationcliffordfloquetdynamics,kapustin2026thermalizationmanybodyclassicalfloquet}. The authors of \cite{G_tschow_2010} demonstrated that fractal-type Clifford QCAs thermalize translation-invariant product states, as well as Pauli stabilizer states, to the tracial state. More recently, these results were generalized in \cite{kapustin2026chaosthermalizationcliffordfloquetdynamics}, and for their classical analogs in \cite{kapustin2026thermalizationmanybodyclassicalfloquet}, to arbitrary spatial dimension, number of qubits per site, and a much broader class of initial states. Here, we continue this research program by analyzing thermalization in quantum many-body systems with non-compact, continuous degrees of freedom under the action of Gaussian QCAs.

 The local probes considered here are local Weyl operators, while the dynamics is given by translation-invariant Gaussian QCAs.  Such automata have two useful features.  First, they preserve locality: the support of a strictly local observable can grow only at finite speed.  Second, they preserve the Weyl relations, so their action on Weyl operators is governed by a linear symplectic map on phase-space labels.  The long-time problem for local Weyl observables can therefore be reduced to the spectral dynamics of a finite-spread symplectic convolution operator.

For clarity we first describe the setting informally.  The lattice is $\La=\IZ^d$, with one bosonic mode at each site.  A local Weyl operator is denoted by $W(\xi)$, where the label $\xi\in V=c_{00}(\La,\IR^2)$ is a finitely supported phase-space vector.  A Gaussian QCA acts on the Weyl operators by
\begin{align}
        \alpha(W(\xi))=e^{i\chi(\xi)}W(\Phi\xi),
\end{align}
where $\Phi:V\to V$ is a translation-invariant finite-spread symplectic map.  After Fourier transform, $\Phi$ is represented by a matrix-valued Laurent-polynomial symbol $A(\theta)$ over the Brillouin torus.  Thus the long-time behavior of local Weyl expectations is controlled by two competing effects: the phase-space norm of $\Phi^k\xi$, and the size of the support of $\Phi^k\xi$.

The local Weyl unitaries generate the $*$-algebra
$\CA_{\mathrm{loc}}$ defined in Section~\ref{section 3}. A state on it is a
normalized positive linear functional. The  initial states considered in this paper are locally normal with respect to
the finite-volume Schr\"odinger representations and have finite particle density.  The latter condition means that there exists $\nu<\infty$ such that, for
every finite region $\Gamma\subset\IZ^d$,
\begin{align}\label{eq:finite_particle_density}
    \omega(N_\Gamma)\leqslant \nu |\Gamma|,
    \qquad
    N_\Gamma=\sum_{x\in\Gamma}a_x^*a_x.
\end{align}
Here $\omega$ is a state while $a_x^*$, $a_x$, and $N_\Gamma$ are the unbounded creation,
annihilation, and number operators in the finite-volume Schr\"odinger
representation. We make sense of the left-hand side of \eqref{eq:finite_particle_density} using Proposition-Definition~\ref{proposition_definition_unbounded}.  

By thermalization we mean convergence of the initial state on $\CA_{\mathrm{loc}}$ to the infinite-temperature Weyl state $\omega_{\infty}$, characterized by
\begin{align}
    \omega_\infty(W(\xi))=\delta_{\xi,0},
    \qquad \xi\in V.
    \label{eq:infinite-temperature-weyl-state}
\end{align}
This prescription defines a state. Indeed, after collecting equal labels, write
$A=\sum_{\xi}c_\xi W(\xi)$ with distinct labels. Then the Weyl relations give
\begin{align}
    \omega_\infty(A^*A)=\sum_{\xi}|c_\xi|^2\geqslant0.
\end{align}
The same relations show that $\omega_\infty$ is tracial on
$\CA_{\mathrm{loc}}$. This state is the analogue of the tracial state in spin systems where the algebra of local observables is generated by the Pauli unitaries.

We prove two versions of the result.  The first uses a global hyperbolicity condition: the spectrum of $A(\theta)$ is uniformly separated from the unit circle for all $\theta\in\IT^d$.  This gives a stable/unstable splitting $P=P_-\oplus P_+$ of the Hilbert phase space $P=\ell^2(\La;\IR^2)$.  We call the GQCA \emph{everyday} when no nonzero strictly local label lies entirely in the stable subspace, equivalently $V\cap P_-=\{0\}$.  Under this condition every nonzero local label has an unstable component and is exponentially expanded by forward time evolution.

\begin{theorem}\label{th:main}
    Let $\alpha$ be an everyday uniformly hyperbolic Gaussian QCA and let $\omega$ be a locally normal state on $\CA_{\mathrm{loc}}$ with finite particle density.  Then, for every nonzero strictly local Weyl label $\xi\in V$,
\begin{align}
        \lim_{k\to\infty}\omega\bigl(\alpha^k(W(\xi))\bigr)=0 .
\end{align}
Equivalently, for every finite linear combination
$A=\sum_{\xi}c_\xi W(\xi)\in\CA_{\mathrm{loc}}$,
\begin{align}
\lim_{k\to\infty}\omega\bigl(\alpha^k(A)\bigr)
    = c_0
    = \omega_\infty(A) .
\end{align}
\end{theorem}

The second version is useful when hyperbolicity is visible only on an open part of the Brillouin torus $\IT^d$.  We call $\Phi$ locally hyperbolic if  $A(\theta_0)$ is hyperbolic at some $\theta_0\in \IT^d$.  We call $\Phi$ \emph{regular} if no nonzero local label is mapped by $\Phi$ to a scalar multiple of a lattice translate of itself\footnote{In particular, this condition implies the absence of gliders in the language of \cite{G_tschow_2010}.}.  Regularity rules out local eigenlabels that could avoid the expanding directions.  In the one-mode uniformly hyperbolic case, regularity is equivalent to everydayness; see Proposition~\ref{prop:fractality_regularity}.

\begin{theorem}\label{th:main_local}
    Let $\alpha$ be a regular locally hyperbolic Gaussian QCA and let $\omega$ be a locally normal state on $\CA_{\mathrm{loc}}$ with finite particle density.  Then, for every nonzero strictly local Weyl label $\xi\in V$,
\begin{align}
        \lim_{k\to\infty}\omega\bigl(\alpha^k(W(\xi))\bigr)=0 .
\end{align}
Equivalently, for every finite linear combination
$A=\sum_{\xi}c_\xi W(\xi)\in\CA_{\mathrm{loc}}$,
\begin{align}
\lim_{k\to\infty}\omega\bigl(\alpha^k(A)\bigr)
    = c_0
    = \omega_\infty(A) .
\end{align}
\end{theorem}
\begin{figure}[htbp]
 \begin{center}
\includegraphics[scale = 0.15]{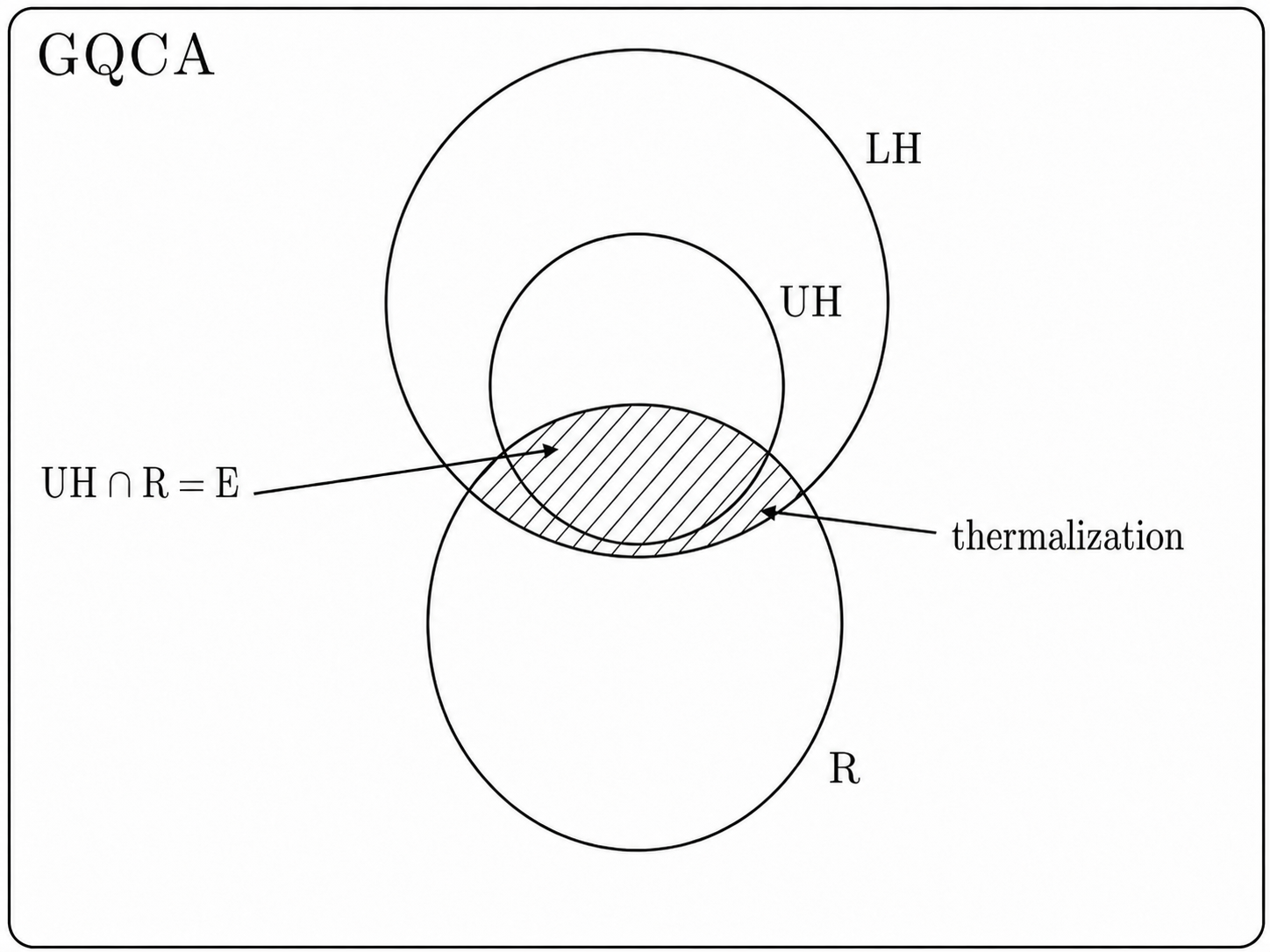}\caption{ $\operatorname{LH}$, $\operatorname{UH}$, $\operatorname{R}$, and $\operatorname{E}$ stand for the local and uniform hyperbolicity, regularity and everydayness, respectively. The shaded region $\operatorname{LH}\cap \operatorname{R}$ contains the GQCAs for which we prove thermalization.}
\label{Fig:diagram}
\end{center}
\end{figure}
We summarize the conditions on thermalizing GQCAs in Fig \ref{Fig:diagram}. 
\begin{remark}[Scope of the theorems]\label{rem:scope-of-theorem}
The conclusion is deliberately formulated at the level of local Weyl observables, equivalently finite linear combinations of local Weyl operators.  Within that class the assumptions on the initial state are weak: the state need not be translation invariant, quasi-free, or Gibbs.  The proof uses no special normal form for the automaton beyond finite spread, translation invariance, and the stated hyperbolicity hypotheses.

The limitation is equally important.  The theorem is not a statement about arbitrary bounded operators in \(\CB(\CH_\Gamma)\), nor about arbitrary unbounded local observables.  The unbounded number operators enter only through the moment condition defining finite particle density.
\end{remark}

\begin{remark}[Multiple bosonic modes per unit cell]\label{rem:Multiple_modes}
Our results are formulated for a single bosonic mode per unit cell. Instead, we could have set the on-site Hilbert space as $\CH=L^2(\IR^r)$. In that case, the setup and the proof of Theorem \ref{th:main} holds through. 
\end{remark}

The analytic input behind both theorems is a many-body version of the quantum Riemann--Lebesgue lemma of Werner \cite{10.1063/1.526310}.  In the present setting, the role of Fourier coefficients is played by expectations of Weyl operators.

\begin{theorem}[Many-body quantum Riemann--Lebesgue estimate]\label{thm:rl-intro}
Let $\omega$ be a locally normal state satisfying  \eqref{eq:finite_particle_density}.  If $0\ne\xi\in V$ and $\Gamma=\supp\xi$, then
\begin{equation}\label{eq:intro-rl}
        \abs{\omega(W(\xi))}
        \leqslant
        \frac{\pi\sqrt{2\nu\abs{\Gamma}+1}}{\norm{\xi}_2} .
\end{equation}
\end{theorem}

The mechanism of the proof is then simple.  Let $\xi_k=\Phi^k\xi$.  The dynamical hypotheses imply $\norm{\xi_k}_2\geqslant c_\xi e^{\gamma k}$ for every nonzero local label $\xi$.  On the other hand, finite spread implies $\abs{\supp\xi_k}\leqslant C_\xi(1+k)^d$.  Substituting these two bounds into \eqref{eq:intro-rl} gives
\begin{align}
    \abs{\omega(W(\Phi^k\xi))}
    \leqslant
    \frac{\pi\sqrt{2\nu C_\xi(1+k)^d+1}}{c_\xi e^{\gamma k}}
    \longrightarrow 0 .
\end{align}
Thus exponential growth in phase-space frequency dominates the polynomial growth of spatial support, and local Weyl observables lose all memory of the initial state.

The paper is organized as follows.  In Section~\ref{section 3} we define the algebraic local Weyl algebra, locally normal states, Gaussian QCAs, and the finite-particle-density condition.  In Section~\ref{section 4} we prove the many-body quantum Riemann--Lebesgue estimate.  In Section~\ref{sec:hyperbolicity_and_regularity} we introduce local and uniform hyperbolicity, everydayness, and regularity, and we prove the algebraic criteria used to check them.  In Section~\ref{sec:thermalization} we combine the dynamical growth estimates with the Riemann--Lebesgue estimate to prove the two thermalization theorems.  Finally, Section~\ref{sec:example} gives an explicit one-dimensional family of examples.

\section{Many-body systems of infinite-dimensional degrees of freedom}\label{section 3}

In this work we analyze many-body quantum systems of infinite-dimensional degrees of freedom arranged in cubic lattices. We consider a lattice $\La=\IZ^d$ and the one-site Hilbert space\footnote{More generally, one can consider $L^2(\IR^r)$ as the on-site Hilbert space. Some of our results generalize to this case straightforwardly.}  $\CH=L^2(\IR)$. Our results concern states and dynamics in infinite volume and for that reason we use the framework of $^*$-algebras of local observables, and view normalized positive linear functionals on them as physical states.

In this work we consider bounded operators, denoted by $\CB=\CB(\CH)$, unbounded operators, and Weyl operators acting on $\CH$, with the latter being the local probes of our interest. We define the on-site ``phase space'' $E_0=\IR^2$ with the standard symplectic form \begin{align}\label{eq:sigma0}
    \sigma_0(\xi_1,\xi_2)=u_1v_2-u_2v_1=\xi_1^T J\xi_2,
    \qquad
    \xi_j=(u_j,v_j),
    \qquad
    J=\begin{pmatrix}0&+1\\-1&0\end{pmatrix}.
\end{align} 
For a finite region $\Gamma\subset\La$, set
\begin{align}
    V_\Gamma:=\bigoplus_{x\in\Gamma}E_0,
    \qquad
    \sigma_\Gamma(\xi_1,\xi_2):=
    \sum_{x\in\Gamma}\sigma_0(\xi_1(x),\xi_2(x))\hspace{3mm} \mbox{and}\hspace{3mm} \CH_\Gamma&=\bigotimes_{x\in \Gamma}\CH\,.
\end{align} 
Define $\CA_\Gamma$ to be the complex $*$-algebra linearly spanned by
Weyl symbols $W_\Gamma(\xi)$, $\xi\in V_\Gamma$, with multiplication and
involution determined by the Weyl relations
\begin{align}\label{eq:Weyl_rels}
    W_\Gamma(\xi_1)W_\Gamma(\xi_2)
    =e^{\frac{i}{2}\sigma_\Gamma(\xi_1,\xi_2)}W_\Gamma(\xi_1+\xi_2),
    \qquad
    W_\Gamma(\xi)^*=W_\Gamma(-\xi).
\end{align} 
For a one-point region, $\CA_{\{x\}}$ is the on-site algebraic Weyl CCR
$*$-algebra.  If $\Gamma_1\subset\Gamma_2$, we include
$\CA_{\Gamma_1}\hookrightarrow\CA_{\Gamma_2}$ by extending labels by zero:
\begin{align}
    W_{\Gamma_1}(\xi)\longmapsto W_{\Gamma_2}(\widetilde \xi),
    \qquad
    \widetilde\xi|_{\Gamma_1}=\xi,
    \quad
    \widetilde\xi|_{\Gamma_2\setminus\Gamma_1}=0.
\end{align}
The algebra of local Weyl observables is the algebraic inductive limit
\begin{align}
    \CA_{\mathrm{loc}}
    :=\bigcup_{\Gamma\Subset\La}\CA_\Gamma.
\end{align}

Throughout the paper, we assume that $\CH$ is the Schr\"odinger representation of $\CA_{\{x\}}$, i.e. $\pi:\CA_{\{x\}}\to \CB(\CH)$ is given by
\begin{align}\label{eq:Weyl_action_R2}
    [\pi(W(u,v))f](x)=e^{i v(x+u/2)}f(x+u),
    \qquad f\in L^2(\IR)\,.
\end{align} 

This representation is regular and irreducible, and by the Stone--von Neumann
theorem it is unique up to unitary equivalence
 among regular irreducible
representations, see e.g. \cite{reed1980methods}.  By Stone's theorem, in this representation there are
unbounded position and momentum operators $\hat q$ and $\hat p$  with
 $[\hat q,\hat p]=i$ on the common core and
  $\pi(W(u,v))=e^{i(u\hat p+v\hat q)}$.
 For finite $\Gamma$, $\CA_{\Gamma}$ acts on $\CH_{\Gamma}$ as the tensor product of Schr\"odinger representations, denoted by $\pi_{\Gamma}$.

A state on $\CA_{\mathrm{loc}}$ is a normalized positive linear functional:
$\omega(1)=1$ and $\omega(a^*a)\geq 0$ for all
$a\in\CA_{\mathrm{loc}}$.  We say that $\omega$ is locally normal, with
respect to the Schr\"odinger net, if for every finite
$\Gamma\Subset\La$ there exists a density operator
$\rho_\Gamma\in\CT^1(\CH_\Gamma)$ such that
\begin{align}\label{eq:local-normality-schrodinger}
\omega(a)=\Tr_{\CH_\Gamma}\!\bigl(\rho_\Gamma\pi_\Gamma(a)\bigr),
    \qquad
    a\in\CA_\Gamma.
\end{align}
Throughout the paper, we suppress $\pi_\Gamma$ and write
$W_\Gamma(\xi)$ both for the abstract Weyl symbol and for its represented
Schr\"odinger operator.

\begin{remark}
    A standard example of a non-normal state on $\CA_{\{x\}}$ is the sharp position state: 
    \begin{align}
        \omega_{q_0}(W(u,v))=\begin{cases}
            e^{ivq_0}\,,\hspace{3mm}u= 0\,,\\
            0\,,\hspace{8mm} \mbox{otherwise}\,.
        \end{cases}
    \end{align}
    Similarly, the sharp momentum state is not normal either.
\end{remark}

The strictly local phase-space labels are given by the finitely-supported sequences from
\begin{align}
    V\coloneqq c_{00}(\La,\IR^2).
\end{align}
For $\xi\in V$ and any finite $\Gamma$ containing $\supp(\xi)$, we write
$W(\xi)$ for $W_\Gamma(\xi|_\Gamma)\in\CA_\Gamma\subset\CA_{\mathrm{loc}}$.
For $\xi_1,\xi_2\in V$, define
\begin{align}
    \sigma(\xi_1,\xi_2)
    :=\sum_{x\in\La}\sigma_0\bigl(\xi_1(x),\xi_2(x)\bigr).
    \label{eq:global-symplectic-form}
\end{align}
The sum is finite because both labels are finitely supported.

We will interchangeably view $V$ as a sequence space or as a space of functions.
It is also convenient to use the Laurent-polynomial group ring
$R:=\IR[u_1^{\pm1},\ldots,u_d^{\pm1}]\cong\IR[\La]$, so that
$V\cong R^2$ is a free module of rank $2$.

Together with the set of labels $V$ for local Weyl operators, we define its
completion with respect to the $\ell^2$-norm
\begin{align}
    P\coloneqq\ell^2(\IZ^d,\IR^2),
    \qquad
    \norm{\xi}_2^2=\sum_x \norm{\xi_x}_{\IR^2}^2.
\end{align}
The same formula for $\sigma$ extends continuously to $P$, since the sum is
absolutely convergent by Cauchy--Schwarz.

\subsection{Gaussian Quantum Cellular Automata}

We now define the class of dynamics analyzed in this paper. Let $T_y$ denote
translation by $y\in \Lambda$ on the label space $V$
\begin{align}
        (T_y \xi)(x) := \xi(x-y),\,\quad x\in \La\,
\end{align}
and let $\tau_y$ denote the lattice translation on local
observables. In particular, on the local Weyl operators,
\begin{align}
    \tau_y(W(\xi)) = W(T_y \xi).
\end{align}
For finite sets $\Gamma_1,\Gamma_2\subset \Lambda$, we write
\begin{align}
    \Gamma_1+\Gamma_2 := \{x+y:x\in \Gamma_1,\ y\in \Gamma_2\}.
\end{align}

\begin{definition}[Gaussian quantum cellular automaton]
A Gaussian quantum cellular automaton, or GQCA, is a $*$-automorphism
$\alpha$ of the algebraic local Weyl algebra $\CA_{\mathrm{loc}}$ for which
there exist a real-linear bijection $\Phi:V\to V$ and a character
$\chi:V\to\IR/2\pi\IZ$ such that
\begin{align}
    \alpha(W(\xi))=e^{i\chi(\xi)}W(\Phi\xi),
    \qquad \xi\in V.
\end{align}
We require the following locality and covariance conditions.
\begin{enumerate}
    \item $\Phi$ and $\Phi^{-1}$ have finite spread: there is a finite set
    $S\subset\La$ such that, for every $\xi\in V$,
    \begin{align}
        \supp(\Phi\xi)\subset\supp(\xi)+S,
        \qquad
        \supp(\Phi^{-1}\xi)\subset\supp(\xi)+S.
    \end{align}
    Equivalently, $\alpha(\CA_\Gamma)\subset\CA_{\Gamma+S}$ and
    $\alpha^{-1}(\CA_\Gamma)\subset\CA_{\Gamma+S}$ for every finite
    $\Gamma\subset\La$.

    \item The pair $(\Phi,\chi)$ is translation invariant:
    \begin{align}
        \Phi T_y=T_y\Phi,
        \qquad
        \chi(T_y\xi)=\chi(\xi),
        \qquad y\in\La,\ \xi\in V.
    \end{align}
    Equivalently,
    \begin{align}
        \alpha\circ\tau_y=\tau_y\circ\alpha
    \end{align}
    on $\CA_{\mathrm{loc}}$.

\end{enumerate}
The Weyl relations imply that $\Phi$ preserves the symplectic form,
\begin{align}
    \sigma(\Phi\xi_1,\Phi\xi_2)=\sigma(\xi_1,\xi_2),
    \qquad \xi_1,\xi_2\in V.
\end{align}
Conversely, because $\CA_{\mathrm{loc}}$ is generated algebraically by the Weyl
symbols, every finite-spread symplectic bijection $\Phi:V\to V$, together with
such a phase character $\chi$, defines a $*$-automorphism by the displayed
formula.
\end{definition}
\textbf{Notation:} The non-scalar part of a GQCA is a finite-spread symplectic map on the
phase-space labels. The scalar phase $e^{i\chi(\xi)}$ will play no role in our
analysis; we therefore suppress the phases in what follows and we refer to $\Phi$ as GQCA.

\vspace{3mm}
Let $R := \mathbb R[u_1^{\pm 1},\ldots,u_d^{\pm 1}]
    \cong \mathbb R[\La]$
be the Laurent-polynomial group ring of the lattice. Under the identification
$V\cong R^{2}$, translation-invariance and finite spread imply that $\Phi$ is represented by a Laurent-polynomial matrix
\begin{align}
    \Phi(u)=\sum_{x\in \Gamma}\Phi_x u^x
    \in \operatorname{Mat}_{2}( R),
    \qquad
    u^x:=u_1^{x_1}\cdots u_d^{x_d},
    \end{align}
where $\Gamma\subset \La$ is finite and
$\Phi_x\in \operatorname{Mat}_{2}(\mathbb R)$. Its action on labels is the
finite convolution
\begin{align}
    (\Phi \xi)(x)
    =
    \sum_{y\in \Gamma}\Phi_y \xi(x-y),
    \qquad x\in\Lambda,\ \xi\in V.
\end{align}

The condition that $\Phi$ preserves $\sigma$ is equivalently
\begin{align}\label{eq:fiber_symplectic_condition}
    \Phi(u^{-1})^T J \Phi(u)=J.
\end{align}
We denote by $\operatorname{Sp}( R,2)$ the group of such Laurent-polynomial symplectic matrices. The inverse is again Laurent-polynomial, since $\Phi(u)^{-1}=J^{-1}\Phi(u^{-1})^T J$.

The same finite convolution formula extends $\Phi$ from the local label space
$V$ to the Hilbert phase space
\begin{align}
    P=\ell^2(\mathbb Z^d,\mathbb R^{2}).
\end{align}
This extension is useful analytically, although a general element of $P$ need
not label a strictly local Weyl operator.

\begin{lemma}
Every $\Phi\in \operatorname{Sp}( R,2)$ acts on
$P=\ell^2(\mathbb Z^d,\mathbb R^{2})$ as a bounded invertible operator.
\end{lemma}

\begin{proof}
Let $T_m$ denote the unitary translation on $P$. Then,
\begin{align}
    \Phi \xi=\sum_{x\in \Gamma}\Phi_x T_x \xi.
\end{align}
Hence
\begin{align}
    \|\Phi \xi\|_2
    \leqslant
    \sum_{x\in \Gamma}\|\Phi_xT_x \xi\|_2
    \leqslant
    \left(\sum_{x\in \Gamma}\|\Phi_x\|\right)\|\xi\|_2.
\end{align}
Thus $\Phi$ is bounded on $P$. Since $\Phi^{-1}$ is also represented by a
Laurent-polynomial matrix, the same argument applies to $\Phi^{-1}$, so the
extension is bounded and invertible.
\end{proof}

\section{Many-body quantum Riemann-Lebesgue lemma }\label{section 4}

In this section we prove estimates on expectation values of local Weyl operators which are essential for the further discussion of thermalization. The classic Riemann-Lebesgue lemma states that the Fourier coefficients of an integrable function vanish at infinity. The quantum Riemann-Lebesgue lemma of R. Werner states that the ``quantum Fourier coefficients'' of a normal state vanish at infinity \cite{10.1063/1.526310}. The main result of this section is a many-body generalization of the quantum Riemann-Lebesgue lemma where expectation values of local Weyl operators play the role of the Fourier coefficients. 

The idea of a proof of the quantum many-body Riemann-Lebesgue lemma is inspired by one of the standard proofs of the the classic Riemann-Lebesgue lemma: take a half-period shift of coordinates in the integral and observe that the difference between the original and shifted integrals goes to zero. In our case, the half-period shift is realized by conjugation with certain Weyl operators, allowing us to suppress the expectation value by the norm of a Weyl label. 

\begin{remark}
Before treating the genuinely many-body situation, let us record the simplest case where the dynamics does not spread the support of a Weyl label but makes its phase-space norm diverge. Then one can use the ordinary finite-dimensional quantum Riemann--Lebesgue lemma of \cite{10.1063/1.526310}.

\begin{lemma}
\label{prop:fixed_support_RL}
Let $\omega$ be a locally normal state on $\CA_{\mathrm{loc}}$ and let
$\xi\in V$.  Put $\xi_k=\Phi^k\xi$.  Assume that there exists a finite region
$\Gamma\subset\IZ^d$ such that
$\operatorname{supp}(\xi_k)\subseteq\Gamma$ for all $k\geqslant0$ and that
$\|\xi_k\|_2\to\infty$ in $\IR^{2|\Gamma|}$.  Then,
\begin{align}
    \lim_{k\to\infty}\omega(W(\xi_k))=0.
    \label{eq:fixed_support_RL_decay}
\end{align}  
\end{lemma}

\begin{proof}
Since all $W(\xi_k)$ belong to the same finite-volume Weyl algebra
$\CA_\Gamma$, local normality gives a density operator $\rho_\Gamma$ on
$\CH_\Gamma$ such that, in the Schr\"odinger representation,
\[
    \omega(W(\xi_k))=
    \Tr_{\CH_\Gamma}(\rho_\Gamma W_\Gamma(\xi_k)).
\]
The right-hand side is the quantum Fourier transform of the trace-class operator
$\rho_\Gamma$ evaluated at $\xi_k$.  Since $\|\xi_k\|\to\infty$, the quantum
Riemann--Lebesgue lemma \cite{10.1063/1.526310} for finitely many oscillators
gives $\operatorname{Tr}_{\CH_\Gamma}(\rho_\Gamma W_\Gamma(\xi_k))\to0$.
\end{proof}

Lemma~\ref{prop:fixed_support_RL} proves decay for an individual
fixed-support orbit whose phase-space norm tends to infinity; it does not by
itself imply full thermalization. For example, a nonspreading hyperbolic
symplectic dynamics also has stable local directions. If
$0\neq\eta\in V$ is such a direction, then $\Phi^k\eta\to0$, and local
normality together with strong continuity of the finite-volume Weyl
representation gives
\begin{align}
    \omega(W(\Phi^k\eta))\longrightarrow1\,.
\end{align}
 When supports grow, the density operators live on larger and larger
tensor factors, so the preceding fixed-volume Riemann--Lebesgue argument is
not uniform. The estimates beginning with
Lemma~\ref{thm: Weyl average bound} replace it by a bound involving both
$\|\Phi^k\xi\|_2$ and $|\operatorname{supp}(\Phi^k\xi)|$.

\vspace{3mm}
\end{remark}

This section is structured as follows. In \S\ref{subsection 4.1}, we introduce a commutator that bounds the expectation value of Weyl operators in locally normal states. In \S\ref{subsection 4.2}, we introduce convention for handling unbounded operators and prove some useful convergences and bounds on states and observables. In \S\ref{subsection 4.3} we connect local particle number to our bounds, followed by  \S\ref{subsection 4.4} where we demonstrate that finite local particle number averages yield trace-class commutators between density operators and Weyl operators. Finally, in \S\ref{subsection 4.5}, we prove that, for bosonic many-body systems with bounded particle density, Weyl operator averages on locally normal states are bounded by the ratio of the QCA spread to the support of the Weyl label. That is to say, in the unstable sector of the QCA's corresponding phase space map, the Weyl label norm grows exponentially, so long-time Weyl averages vanish.

\subsection{Bounded Weyl averages}\label{subsection 4.1}

For $\xi=(u,v)\in V_\Gamma$, let $R_\Gamma(\xi)$ denote the self-adjoint
generator of the strongly continuous one-parameter group
$t\mapsto W_\Gamma(t\xi)$. On the standard common invariant core,
\begin{align}
    R_\Gamma(\xi)
    =\sum_{x\in\Gamma}\bigl(u(x)\hat p_x+v(x)\hat q_x\bigr),
    \qquad
    W_\Gamma(\xi)=e^{iR_\Gamma(\xi)}.
\end{align}

\begin{definition}
    Let $\Gamma \subset \bZ^d$ be a finite subset and $\rho_\Gamma$ a density operator on $\mathcal{H}_\Gamma$. For $r\geqslant 0$
    \begin{align}
        \Omega_{\rho_\Gamma}(r) \coloneqq \sup_{\lVert{\xi}\rVert \leqslant r} \lVert W_\Gamma(\xi)\rho_\Gamma W_\Gamma(\xi)^* - \rho_\Gamma\rVert_1.
    \end{align}
\end{definition}

Let
\begin{align}
    J_\Gamma:=\bigoplus_{x\in\Gamma}J.
\end{align}
For $\xi\in\IR^{2|\Gamma|}\setminus\{0\}$, define the sign-flipping
coordinate
\begin{align}
    s_\xi:=\frac{\pi J_\Gamma\xi}{\|\xi\|^2}.
\end{align}
Then, $
    \|s_\xi\|=\frac{\pi}{\|\xi\|}$ and $
    \sigma_\Gamma(s_\xi,\xi)=\pi$,
and the Weyl relations give
\begin{align}
    W_\Gamma(s_\xi)^*W_\Gamma(\xi)W_\Gamma(s_\xi)
    =-W_\Gamma(\xi).
\end{align}

In what follows, let $\lVert \cdot\rVert_1$ denote the trace norm, $\lVert T\rVert_1 =\Tr(\sqrt{T^*T})$ and $\CT^1$ denote the ideal of trace-class operators $\CT^1\coloneqq  \{ A\in \mathcal{B(H)}\,|\, \|A\|_1<  \infty\}$.
\begin{lemma}\label{thm: Weyl average bound}
    Let $\alpha$ be a Gaussian QCA with induced phase-space map $\Phi$, $\omega$ locally normal, and let 
    \begin{align}
        \xi_0 \in V\setminus\{0\}, \quad \xi_k := \Phi^k \xi_0, \quad \Gamma_k := \mathrm{supp}(\xi_k).
    \end{align} 
    Then, for every $k\geqslant 0$ with $\xi_k \neq 0$, 
    \begin{align}
        |\omega(W(\xi_k))|\leqslant \frac{1}{2} \Omega_{\rho_{\Gamma_k}}\left(\frac{\pi}{\lVert \xi_k \rVert}\right)
    \end{align}
    with $\rho_{\Gamma_k}$ the density operator corresponding to $\omega$ on $\CA_{\Gamma_k}$.
\end{lemma}
\begin{proof}
    Let $s_k$ denote the sign-flipping coordinate associated to $\xi_k$. Then,
    \begin{align}
        2\omega(W(\xi_k)) &= 2\Tr \big[\rho_{\Gamma_k}W_{\Gamma_k}(\xi_k)\big] = \Tr\bigg[\big(\rho_{\Gamma_k} - W(s_{\xi_k})\rho_{\Gamma_k}W(s_{\xi_k})^*\big)W(\xi_k)\bigg].
    \end{align}
    Recall that, for $T\in \mathcal{T}^1$ and $B$ a bounded operator, Theorem 2.8 of \cite{simon2005trace} gives the Schatten-H\"older inequality
    \begin{align}
        \lVert TB \rVert_1 \leqslant \lVert T\rVert_1 \lVert B\rVert.
    \end{align}
    Consequently,
    \begin{align}
        |\omega(W(\xi_k))| \leqslant \frac{1}{2} \lVert \rho_{\Gamma_k} - W(s_k)\rho_{\Gamma_k} W(s_k)^*\rVert_1 \leqslant \frac{1}{2} \Omega_{\rho_{\Gamma_k}} \left(\frac{\pi}{\lVert \xi_k\rVert} \right).
    \end{align}
\end{proof}
\subsection{Some structural results on positive local observables}\label{subsection 4.2}

We outline certain structural results on positive observables that are needed to prove our claims. A standard reference on the spectral theory of unbounded operators is \cite{reed1980methods}, Chapter VIII. We denote the domain of an unbounded operator $A$ by $D(A)$. In this paper we use only densely-defined unbounded operators so that the adjoints always exist. We use the convention that the (densely defined) unbounded operator $A$ is self-adjoint if $D(A)=D(A^*)$ and $A^*=A$. 

The following definitions are standard, see \cite[\S VIII.6]{reed1980methods}, \cite[Ch.~VI]{kato1995perturbation}.

\begin{definition}\label{def:unbounded_expectation_values}
Let $\rho$ be a density operator, then for any positive self-adjoint operator $A$ (possibly unbounded), we use the spectral theorem to define the expectation value of $A$ with respect to $\rho$ as

\begin{align}
    \omega_\rho(A) \coloneqq \sup_{n \in \IN} \Tr \left(A^{(n)}\rho \right) \in [0,\infty],
\end{align}
for the spectral truncation
\begin{align}
    A^{(n)}:=A P_A([0,n])
    =\int_{[0,n]}\lambda\,dP_A(\lambda),
\end{align}
where $P_A$ is the projection-valued spectral measure of $A$.

\end{definition}
\begin{proposition-definition}\label{proposition_definition_unbounded}
    Let $\rho = \sum_{j\geqslant 1}\lambda_j \psi_j \psi_j^*$ be the spectral decomposition of a density operator on $\mathcal{H}$. We choose the convention that $\left\lVert A^{1/2}\psi_j \right\rVert  =\infty $ when $\psi_j \notin D(A^{1/2})$. For $A$ a positive self-adjoint operator on $\CH$:
    \begin{align}
        \omega_\rho(A) = \sum_{j\geqslant 1} \lambda_j \left\lVert A^{1/2}\psi_j \right\rVert^2 \in [0, \infty]\,.
    \end{align}
    \label{prop: expectation of positive unbounded op}
\end{proposition-definition}
\begin{proof}
    We have for each $n$, a well-defined expression
    \begin{align}
        \Tr(A^{(n)}\rho) = \sum_{j \geqslant 1} \lambda_j \langle\psi_j|A^{(n)}\psi_j\rangle.
    \end{align}
    For each fixed $j$,
    \begin{align}
        \langle\psi_j|A^{(n)}\psi_j\rangle \overset{n\rightarrow\infty}{\longrightarrow}  \left\lVert A^{1/2}\psi_j \right\rVert^2 
    \end{align}
    from below. Therefore, monotone convergence on each summand yields
    \begin{align}
        \sup_n \Tr \left(A^{(n)}\rho \right) = \sum_{j \geqslant1} \lambda_j \left\lVert A^{1/2}\psi_j \right\rVert^2 .
    \end{align}
\end{proof}
\begin{corollary}
    Let $A_1, A_2$ be positive self-adjoint operators on $\mathcal{H}$. Suppose that $D(A_2^{1/2}) \subseteq D(A_1^{1/2})$ and 
    \begin{align}
        \left\lVert A_1^{1/2}\psi \right\rVert^2 \leqslant \left\lVert A_2^{1/2}\psi \right\rVert^2 , \;\;\; \psi \in D(A_2^{1/2}).
    \end{align}
    Then, for every density operator $\rho$ with $\omega_\rho(A_2) < \infty$,
    \begin{align}
        \omega_\rho(A_1)  \leqslant \omega_\rho(A_2).
    \end{align}
\end{corollary}
\begin{proof}
    The assumption of finiteness of $\omega_\rho(A_2)$ implies that for all $\psi_j$ in the spectral decomposition of $\rho$, $\psi_j \in D(A_2^{1/2})$, by Proposition-Definition \ref{prop: expectation of positive unbounded op}. The statement immediately follows from term-by-term bounding. 
\end{proof}

We now make sense of the commutator $[A,\rho]$ when $A$ is self-adjoint
and unbounded. The product $A\rho$ is defined on all of $\CH$ provided that
$\operatorname{ran}(\rho)\subseteq D(A)$, whereas $\rho A$ is initially
defined only on $D(A)$.

\begin{definition}
Let $A$ be self-adjoint and let $\rho\in\CT^1(\CH)$. Assume that
$A\rho\in\CT^1(\CH)$ and that the densely defined operator
$\rho A:D(A)\to\CH$ admits a trace-class extension. We denote this unique
extension again by $\rho A$ and define
\begin{align}
    [A,\rho]:=A\rho-\rho A\in\CT^1(\CH).
\end{align}
\end{definition}
Proposition~\ref{prop: finite second moments} below shows that these
conditions follow from a finite second moment.

\begin{lemma}\label{lemma: trace class convergence}
Let $B_n,B\in\CB(\CH)$ and suppose that $B_n\to B$ $*$-strongly. Then, for
every $T\in\CT^1(\CH)$,
\begin{align}
    \|B_nT-BT\|_1\longrightarrow0,
    \qquad
    \|TB_n-TB\|_1\longrightarrow0.
\end{align}
Let $I$ be a compact interval. Suppose in addition that
$B_n(t),B(t)$ are bounded, uniformly in $t$ and $n$,
and, $B_n\to B$ $*$-strongly uniformly in $t$. 
Then
\begin{align}
    \sup_{t\in I}\|B_n(t)T-B(t)T\|_1 \longrightarrow0,\quad 
    \sup_{t\in I}\|TB_n(t)-TB(t)\|_1
        &\longrightarrow0.
\end{align}
\end{lemma}
\begin{proof}
Let $\epsilon>0$ and choose a finite-rank operator $F$ such that
$\|T-F\|_1<\epsilon$. Then
\begin{align}
    \|B_nT-BT\|_1
    \leqslant
    (\|B_n\|+\|B\|)\epsilon+\|(B_n-B)F\|_1.
\end{align}
Writing a singular-value decomposition
$F=\sum_{j=1}^m s_j\psi_j\phi_j^*$ gives
\begin{align}
    \|(B_n-B)F\|_1
    \leqslant
    \sum_{j=1}^m
    s_j\|(B_n-B)\psi_j\|\,\|\phi_j\|
    \longrightarrow0.
\end{align}
Strong convergence implies $\sup_n\|B_n\|<\infty$, so the first
conclusion follows. For left multiplication, apply the same argument to
\begin{align}
    \|TB_n-TB\|_1
    =\|B_n^*T^*-B^*T^*\|_1.
\end{align}
The parametrized statement follows from the same finite-rank approximation,
using the uniform operator bound and the uniform convergence.
\end{proof}

\begin{lemma}
Let $A$ be self-adjoint and let
\begin{align}
    A_n:=A P_A([-n,n])
    =\int_{[-n,n]}\lambda\,dP_A(\lambda)
\end{align}
be its symmetric spectral truncation. Then, for every $\tau>0$,
$e^{itA_n}\to e^{itA}$ strongly, uniformly for $t\in[-\tau,\tau]$.
\end{lemma}

\begin{proof}
    Fix $\psi \in \mathcal{H}$ and let $\mu_\psi$ be the scalar spectral measure of $A$ associated with $\psi$, that is, $\mu_\psi: \mathscr{B}(\IR)\rightarrow \IR$ with $\mu_\psi(I) = \langle\psi | P_A(I)\psi\rangle$ for $P_A$ the spectral projector of $A$ and $\mathscr{B}(\IR)$ the Borel $\sigma$-algebra of $\IR$. We apply the spectral decomposition and obtain 
      \begin{align}
        \sup_{|t| \leqslant \tau} \norm{(e^{itA_n} - e^{itA})\psi}^2\leqslant \sup_{|t| \leqslant \tau}\int_\mathbb{R}  |e^{it\lambda\mathbb{I}_{[-n,n]}(\lambda)} - e^{it\lambda}|^2 \,d\mu_\psi(\lambda)\leqslant \int_\mathbb{R} \sup_{|t| \leqslant \tau} |e^{it\lambda\mathbb{I}_{[-n,n]}(\lambda)} - e^{it\lambda}|^2 \,d\mu_\psi(\lambda).
    \end{align}
    Note that $\mu_{\psi}$ is a finite measure. Since, for all $t$,  $|e^{it\lambda\mathbb{I}_{[-n,n]}(\lambda)} - e^{it\lambda}|\leqslant 2$, dominated convergence gives
    \begin{align}
        \lim_{n\rightarrow \infty} \int_\mathbb{R} \sup_{|t| \leqslant \tau} |e^{it\lambda\mathbb{I}_{[-n,n]}(\lambda)} - e^{it\lambda}|^2 \,d\mu_\psi(\lambda) = \int_\mathbb{R} \lim_{n\rightarrow \infty}\sup_{|t| \leqslant \tau} |e^{it\lambda\mathbb{I}_{[-n,n]}(\lambda)} - e^{it\lambda}|^2 \,d\mu_\psi(\lambda) =0.
    \end{align}

\end{proof}
For the following proposition, we will use Bochner theory to bound trace  norms, in particular, that Lebesgue theory passes through to measurable functions valued in Banach spaces. A standard introduction to the subject is Chapter V of \cite{Yosida78fanalysis}.
\begin{proposition}
    Let $A$ be a self-adjoint operator and let $\rho$ be a density operator. Assume that $A\rho\in\mathcal{T}^1$ and that $\rho A:D(A)\to\CH$ admits a trace-class extension, denoted again by $\rho A$. Then, the map $F(t) \coloneqq e^{iAt}\rho e^{-iAt}$ is continuously differentiable in trace norm, and
    \begin{align}
        F'(t) = ie^{iAt}[A,\rho ]e^{-iAt}.
    \end{align}
    Equivalently,
    \begin{align}
        e^{iAt}\rho e^{-iAt} - \rho = i\int_0^te^{iAs}[A,\rho] e^{-iAs}ds,
    \end{align}
    a Bochner integral in $\mathcal{T}^1$. In particular, 
    \begin{align}
        \norm{e^{itA}\rho e^{-itA} - \rho}_1 \leqslant |t|\,\norm{[A,\rho]}_1.
    \end{align}
\end{proposition}
\begin{proof}

    Since
    \begin{align}
        A_n\rho=P_A([-n,n])A\rho,
        \qquad
        \rho A_n=(\rho A)P_A([-n,n]),
    \end{align}
    and $P_A([-n,n])\to\id$ strongly,
    Lemma~\ref{lemma: trace class convergence} gives
    \begin{align}
        \norm{[A_n,\rho] - [A,\rho]}_1 \to 0.
    \end{align}

    For each $n$, the generator $A_n$ is bounded, so the usual bounded-operator calculus gives a trace norm $C^1$ map $F_n(t) := e^{itA_n}\rho e^{-itA_n}$ with derivative $F_n'(t) = ie^{itA_n }[A_n,\rho]e^{-itA_n}$, and therefore
    \begin{align}
        F_n(t) - \rho = i\int_0^t e^{isA_n}[A_n,\rho]e^{-isA_n}\,ds.
    \end{align}
    Let $\tau>0$. The same uniform strong convergence holds with $t$ replaced by $-t$, so $e^{itA_n}\to e^{itA}$ uniformly $*$-strongly for $|t|\leqslant\tau$. Lemma~\ref{lemma: trace class convergence} therefore gives
    \begin{gather}
        \sup_{|t| \leqslant \tau} \norm{e^{itA_n}\rho e^{-itA_n} - e^{itA}\rho e^{-itA}}_1 \to 0, \\
        \sup_{|t| \leqslant \tau} \norm{e^{itA_n}[A_n,\rho]e^{-itA_n} - e^{itA}[A,\rho]e^{-itA}}_1 \to 0.
    \end{gather}

    Passing to the limit uniformly on compact $t$-intervals gives
    \begin{align}
        e^{itA}\rho e^{-itA} - \rho
        = i\int_0^t e^{isA}[A,\rho]e^{-isA}\,ds.
    \end{align}
    The integrand is trace-norm continuous, so the displayed identity proves
    the asserted continuous differentiability; the triangle inequality gives
    the desired bound.

\end{proof}
Now, in the context of our physical system, we get the following bound.
\begin{corollary} \label{cor: commutator bound}
    Let $\Gamma \subset \mathbb{Z}^d$ be finite, $\rho_\Gamma$ be a density operator on $\mathcal{H}_\Gamma$, and let $\xi \in \IR^{2|\Gamma|}$ satisfy $\norm{\xi} = 1$. Assume that $R_\Gamma(\xi)\rho_\Gamma \in \mathcal{T}^1(\mathcal{H}_\Gamma)$ and $\rho_\Gamma R_\Gamma(\xi) \in \mathcal{T}^1(\mathcal{H}_\Gamma)$. Then for every $t \in \IR$,
    \begin{align}
        \norm{W_\Gamma(t\xi)\rho_\Gamma W_\Gamma(t\xi)^* - \rho_\Gamma}_1 \leqslant |t|\,\norm{[R_\Gamma(\xi),\rho_\Gamma]}_1.
    \end{align}
\end{corollary}
\begin{proof}
    Apply previous proposition to $R_\Gamma(\xi)$ and $W_\Gamma$ and the claim follows. 
\end{proof}
\subsection{Local number operators}\label{subsection 4.3}
To achieve our many-body Riemann--Lebesgue lemma, we control Weyl operator
averages by particle number in the finite-volume Schr\"odinger representations.
Although $\CA_\Gamma$ is the algebraic Weyl $*$-algebra, its locally normal
representatives act on $\CH_\Gamma$, where we have the usual unbounded ladder
operators
\begin{align}
    a_x := \frac{\hat{q}_x + i\hat{p}_x}{\sqrt{2}},
  \qquad
  N_\Gamma := \sum_{x \in \Gamma} a_x^* a_x.
\end{align}
The operator $N_\Gamma$ is not an element of $\CA_\Gamma$; it is an unbounded
operator affiliated with the Schr\"odinger representation and is used only to
impose the finite-density moment condition on locally normal states.

\begin{definition}
    We define the core $\mathcal{D}_{fin, \Gamma}$ of the local number operator $N_\Gamma$, as the finite span of number eigenvectors of the family $\{a^*_xa_x\}_{x\in \Gamma}$.
    \begin{align}
        \CD_{fin,\Gamma}\coloneqq \operatorname{span}_{\IC}\{\psi \in \CH_{\Gamma}\,|\forall x \in \Gamma, \;\,a_x^*a_x \psi=n_x \psi\,\quad \mbox{for some}\,n_x\in \IN\}\,.
    \end{align}
\end{definition}
It is a subset of $D(N_\Gamma)$ such that its graph closure is the graph over the domain, that is, $\overline{\{(x, N_\Gamma x)\,| \,x\in \mathcal{D}_{fin, \Gamma} \}}= \{(x, N_\Gamma x)\,|\, x\in D(N_\Gamma)\}$. Working over the core allows one to avoid certain technical details with respect to the domains of unbounded operators in our next proposition, as it is closed under any finite string of ladder operators. 

\begin{proposition} \label{prop: quadratic forms bound}
    Let $\Gamma \subset \IZ^d$ be finite and let $\xi \in \IR^{2|\Gamma|}$ satisfy $\norm{\xi} = 1$. Then  the following holds:
    \begin{enumerate}[label = (\roman*)]
        \item On $\mathcal{D}_{fin,\Gamma}$ one has
    \begin{align}
        R_\Gamma(\xi)^2 \leqslant 2N_\Gamma + 1
    \end{align}
    as quadratic forms.
        \item $D(N_\Gamma^{1/2}) \subseteq D(R_\Gamma(\xi))$ and for every $\psi \in D(N_\Gamma^{1/2})$,
    \begin{align}
         \norm{R_\Gamma(\xi)\psi}^2 \leqslant 2\langle\psi|N_\Gamma\psi\rangle + \norm{\psi}^2,
    \end{align}
    and so for $\rho_\Gamma$ a density operator with $\omega_{\rho_\Gamma}(N_\Gamma)<\infty$,
    \begin{align}
        \omega_{\rho_\Gamma}(R_\Gamma(\xi)^2)\leqslant 2\omega_{\rho_\Gamma}(N_\Gamma)+1.
    \end{align}
    \end{enumerate} 
\end{proposition}
\begin{proof}
    \mbox{}
    \begin{enumerate}[label = (\roman*)]
        \item Write $\xi = (u, v)$ and define $\gamma_x := v_x - iu_x$ such that $\sum_{x \in \Gamma}|\gamma_x|^2 = 1$. On $\mathcal{D}_{fin,\Gamma}$ we have
    \begin{align}
        R_\Gamma(\xi) = \frac{1}{\sqrt{2}}\sum_{x \in \Gamma}(\gamma_x a_x + \overline{\gamma}_x a_x^*).
    \end{align}
    Define the composite annihilation operator,
    \begin{align}
        b = \sum_{x\in \Gamma} \gamma_xa_x,
    \end{align}
    we have $[b,b^*]=1$ and on $\mathcal{D}_{fin, \Gamma}$,
    \begin{align}
        R_\Gamma(\xi)= \frac{b+b^*}{\sqrt{2}}.
    \end{align}
    For $\psi \in \mathcal{D}_{fin, \Gamma}$, 
    \begin{align}
        \norm{b\psi}^2 = \Bigl\|\sum_{x \in \Gamma}\gamma_x a_x\psi\Bigr\|^2
  \leqslant \sum_{x \in \Gamma}|\gamma_x|^2 \cdot \sum_{x \in \Gamma}\norm{a_x\psi}^2
  = \langle\psi|N_\Gamma\psi\rangle,
    \end{align}
    so $b^*b \leqslant N_\Gamma$ as quadratic forms on $\mathcal{D}_{fin,\Gamma}$. We have
    \begin{align}
         0 \leqslant (b -& b^*)^*(b - b^*) = bb^* + b^*b - b^2 - b^{*2}\nonumber \\
        & \implies b^2 + b^{*2} \leqslant bb^* + b^*b
    \end{align}
    which we can use to say that
    \begin{align}
        (b + b^*)^2 &= b^2 + b^{*2} + bb^* + b^*b \leqslant 2(bb^* + b^*b)= 2(2b^*b + 1)\leqslant 2(2N_\Gamma + 1) \nonumber \\
        \implies \;\;\;\;\;\;& R_\Gamma(\xi)^2 \leqslant 2N_\Gamma + 1
    \end{align}
    as quadratic forms on $\mathcal{D}_{fin,\Gamma}$. \\

    \item For $\psi \in D(N_\Gamma^{1/2})$ and $P_{N_\Gamma,[0,n]}$ the spectral projector onto $[0,n]$, we have that $P_{N_\Gamma,[0,n]}\psi \in \mathcal{D}_{fin,\Gamma}$ with 
    \begin{align}
        P_{N_\Gamma,[0,n]}\psi \overset{n \rightarrow \infty}{\longrightarrow} \psi, \quad N_\Gamma ^{1/2}P_{N_\Gamma,[0,n]}\psi \overset{n \rightarrow \infty}{\longrightarrow} N_{\Gamma} ^{1/2}\psi.
    \end{align}
    Applying the quadratic form inequality from above,
    \begin{align}
        \norm{R_\Gamma(\xi)\left(P_{N_\Gamma,[0,n]}-P_{N_\Gamma,[0,m]}\right)\psi}^2
        &\leqslant
        2\left\langle
          \left(P_{N_\Gamma,[0,n]}-P_{N_\Gamma,[0,m]}\right)\psi,
          N_\Gamma\left(P_{N_\Gamma,[0,n]}-P_{N_\Gamma,[0,m]}\right)\psi
        \right\rangle \nonumber\\
        &\qquad+
        \norm{\left(P_{N_\Gamma,[0,n]}-P_{N_\Gamma,[0,m]}\right)\psi}^2,
    \end{align}
    which goes to zero as $m,n \rightarrow \infty$, so $R_\Gamma(\xi)P_{N_\Gamma,[0,n]}\psi$ is Cauchy in $\mathcal{H}_\Gamma$. Since $R_\Gamma(\xi)$ is self-adjoint, it is closed, and thus $\psi \in D(R_\Gamma(\xi))$ with $R_\Gamma(\xi)P_{N_\Gamma,[0,n]}\psi \overset{n \rightarrow \infty}{\longrightarrow} R_\Gamma(\xi)\psi$. Therefore, for $\psi \in D(N_\Gamma^{1/2})$,
    \begin{align}
        \norm{R_\Gamma(\xi)\psi}^2 \leqslant 2\langle \psi| N_\Gamma \psi\rangle + \norm{\psi}^2. 
    \end{align}
    Lastly, we consider the spectral decomposition $\rho_\Gamma =\sum_j \lambda_j \phi_j\phi_j^*$. Since $\omega_{\rho_\Gamma}(N_\Gamma)< \infty$, Proposition-Definition \ref{prop: expectation of positive unbounded op} implies that for every $\lambda_j>0,$ $\phi_j \in D(N_\Gamma^{1/2})$ and thus
    \begin{align}
        \omega_{\rho_\Gamma}(R_\Gamma(\xi)^2)\leqslant 2\omega_{\rho_\Gamma}(N_\Gamma)+1.
    \end{align}
    \end{enumerate}
\end{proof}
\subsection{Finite second moments imply trace-class commutators}\label{subsection 4.4}
Proposition \ref{prop: quadratic forms bound} establishes an upper bound on the second moment of our Weyl generators. We now show that finite second moments imply trace-class commutators.

\begin{proposition}\label{prop: finite second moments}
Let $\rho$ be a density operator and let $A$ be self-adjoint on $\CH$. If
$\omega_\rho(A^2)<\infty$, then $\operatorname{ran}(\rho)\subseteq D(A)$,
the operator $A\rho$ is trace class, and the densely defined operator
$\rho A:D(A)\to\CH$ admits a trace-class extension, denoted again by
$\rho A$. Moreover,
\begin{align}
    \|A\rho\|_1=\|\rho A\|_1
    \leqslant\sqrt{\omega_\rho(A^2)}.
    \label{eq:finite-second-moment-trace-bound}
\end{align}
Consequently,
\begin{align}
    \|[A,\rho]\|_1
    \leqslant2\sqrt{\omega_\rho(A^2)}.
\end{align}
\end{proposition}
\begin{proof}
Write
\begin{align}
    \rho=\sum_{n\geqslant1}\lambda_n\phi_n\phi_n^*,
\end{align}
where only eigenvectors with $\lambda_n>0$ are listed.
Proposition-Definition~\ref{prop: expectation of positive unbounded op}
gives
\begin{align}
    \sum_{n\geqslant1}\lambda_n\|A\phi_n\|^2
    =\omega_\rho(A^2)<\infty.
\end{align}
Hence
\begin{align}
    \sum_{n\geqslant1}\lambda_n\|A\phi_n\|
    &\leqslant
    \left(\sum_n\lambda_n\right)^{1/2}
    \left(\sum_n\lambda_n\|A\phi_n\|^2\right)^{1/2}
    \nonumber=\sqrt{\omega_\rho(A^2)}.
    \label{eq:first-moment-summability}
\end{align}
Therefore the finite-rank operators
\begin{align}
B_N:=\sum_{n=1}^N\lambda_n(A\phi_n)\phi_n^*
\end{align}
converge in trace norm to an operator $B\in\CT^1(\CH)$, with
\begin{align}
    \|B\|_1\leqslant\sqrt{\omega_\rho(A^2)}.
\end{align}
Let $\rho_N:=\sum_{n=1}^N\lambda_n\phi_n\phi_n^*.$
For every $\psi\in\CH$,
\begin{align}
\rho_N\psi&\longrightarrow\rho\psi\,,\quad A\rho_N\psi=B_N\psi \longrightarrow B\psi.
\end{align}
Since $A$ is closed, it follows that $\rho\psi\in D(A)$ and
$A\rho\psi=B\psi$. Thus
\begin{align}
    A\rho=B\in\CT^1(\CH).
\end{align}
Moreover, on $D(A)$,
\begin{align}
    \rho A\subseteq(A\rho)^*=B^*.
\end{align}
Thus $B^*$ is the trace-class extension of $\rho A$, and
\begin{align}
    \|\rho A\|_1=\|B^*\|_1=\|B\|_1=\|A\rho\|_1.
\end{align}
The commutator estimate follows from the triangle inequality.
\end{proof}

\begin{corollary}\label{cor:number-commutator-bound}
    Let $\Gamma \subset \IZ^d$ be finite, $\rho_\Gamma$ be a density operator on $\mathcal{H}_\Gamma$, and $\xi \in \IR^{2|\Gamma|}$ satisfy $\norm{\xi} = 1$. If $\omega_{\rho_\Gamma}(N_\Gamma) < \infty$, then
    \begin{align}
        R_\Gamma(\xi)\rho_\Gamma \in \mathcal{T}^1(\mathcal{H}_\Gamma), \qquad \rho_\Gamma R_\Gamma(\xi) \in \mathcal{T}^1(\mathcal{H}_\Gamma),
    \end{align}
    and
    \begin{align}
        \norm{[R_\Gamma(\xi), \rho_\Gamma]}_1 \leqslant 2\sqrt{2\omega_{\rho_\Gamma}(N_\Gamma)+1}.
    \end{align}
\end{corollary}
\begin{proof}
Applying Proposition~\ref{prop: finite second moments} to $R_\Gamma(\xi)$, followed by Proposition~\ref{prop: quadratic forms bound}, yields the desired result.
\end{proof}
\subsection{Weyl operator decay}\label{subsection 4.5}
Equipped with the previous results, we are now prepared to bound the norm of Weyl operator averages on locally normal states by local particle number.
\begin{corollary}\label{cor: bounding by commutator}
    Under the assumptions of Lemma \ref{thm: Weyl average bound}, assume that for each $k$ and $\xi\in \IR^{2|\Gamma_k|}$, the commutator $[R_{\Gamma_k}(\xi), \rho_{\Gamma_k}]$ is well defined as above, then
    \begin{align}
        |\omega(W(\xi_k))| \leqslant \frac{\pi}{2}\sup_{\norm{\xi}=1}\frac{\norm{[R_{\Gamma_k}(\xi), \rho_{\Gamma_k}]}_1}{\norm{\xi_k}}.
    \end{align}
\end{corollary}
\begin{proof}
    By Corollary \ref{cor: commutator bound}, we have
    \begin{align}
        \norm{W_\Gamma(t\xi)\rho_\Gamma W_\Gamma(t\xi)^* - \rho_\Gamma}_1 \leqslant |t| \sup_{\norm{\xi}=1} \norm{[R_\Gamma(\xi),\rho_\Gamma]}_1.
    \end{align}
   Now take the supremum over $|t|\leqslant \frac{\pi}{\norm{\xi_k}}$ and apply Lemma~\ref{thm: Weyl average bound}.
\end{proof}

We are now equipped to prove the main Riemann-Lebesgue estimate:
\begin{proof}[Proof of Theorem \ref{thm:rl-intro}]
    Apply Corollary \ref{cor: bounding by commutator} and then Corollary \ref{cor:number-commutator-bound} on the region $\Gamma$. The finite-density estimate $\omega(N_\Gamma)\leqslant \nu|\Gamma|$ gives the desired result.
\end{proof}
\begin{remark}
Let us consider a simple example of locally-normal state without uniform particle-density bound.
Let $\{\varphi_n\}_{n\geqslant0}$ be the one-site
number basis, and choose $n_x=1+|x|^2$. The  family
 \begin{align}
        \rho_\Gamma
        =
        \bigotimes_{x\in\Gamma}
        |\varphi_{n_x}\rangle\langle \varphi_{n_x}|
   \end{align}
defines a locally normal state on $\mathcal A_{\mathrm{loc}}$. On the other hand, $\omega_{\rho_\Gamma}(N_\Gamma)=\sum_{x\in\Gamma}n_x.$
Therefore,
 \begin{align}
        \sup_{\Gamma\Subset\Lambda}
        \frac{\omega_{\rho_\Gamma}(N_\Gamma)}{|\Gamma|}
        \geqslant
        \sup_{x\in\Lambda} n_x=\infty.
   \end{align}
\end{remark}

\begin{corollary}\label{cor: gqca bounds to zero}
    Let $\alpha$ be a Gaussian QCA with induced phase-space map $\Phi$, let $\omega$ be a locally normal state, and 
    \begin{align}
        \xi_0 \in V\setminus\{0\},
  \quad \xi_k := \Phi^k \xi_0,
  \quad \Gamma_k := \mathrm{supp}(\xi_k).
    \end{align}
    Assume that there exists a $\nu$ such that, for every finite $\Gamma \subset \IZ^d$,
    \begin{align}
        \omega_{\rho_\Gamma}(N_\Gamma) \leqslant \nu |\Gamma|,
    \end{align}
    and that there are constants $c>0, \lambda>1$ and $k_0 \in \IN$ such that, for $k \geq k_0$,
    \begin{align}
        \norm{\xi_k} \geqslant c\lambda^k.
    \end{align}
    Then,
    \begin{align}
        |\omega(W(\xi_k))| \leqslant \frac{\pi}{c\lambda^k}\sqrt{2\nu C(1+k)^d + 1}
    \end{align}
    for some $C>0$.
\end{corollary}
\begin{proof}
    By theorem \ref{thm:rl-intro}, 
    \begin{align}
        |\omega(W(\xi_k))| \leqslant \frac{\pi\sqrt{2\nu|\Gamma_k|+1}}{\norm{\xi_k}}.
    \end{align}
    Let $S\subset\mathbb Z^d$ be a finite spread set for $\Phi$. Then
\begin{align}
    \Gamma_k\subset \Gamma_0+kS .
\end{align}
Since $S$ is finite, there is a constant $C>0$, depending only on
$\Gamma_0$ and $S$, such that
\begin{align}
    |\Gamma_k|\leqslant C(1+k)^d .
\end{align}
Therefore,
    \begin{align}
         |\omega(W(\xi_k))| \leqslant \frac{\pi}{c\lambda^k}\sqrt{2\nu C(1+k)^d + 1} \overset{k\rightarrow \infty}{\longrightarrow} 0.
    \end{align}
\end{proof}

\section{Hyperbolicity and regularity of GQCA}\label{sec:hyperbolicity_and_regularity}
In this section we introduce local and uniform hyperbolicity, and regularity of GQCAs. We discuss how combinations of these conditions lead to the dynamics thermalizing locally normal states on $\CA_{\mathrm{loc}}$ with finite particle density. 

\subsection{Uniform hyperbolicity}

The complexified Hilbert space $P_{\IC}\coloneqq P\otimes_{\IR}\IC$ is Fourier dual to $L^2(\IT^d,\IC^{2})$. The dual to $a\in P_{\IC}$ is 
\begin{align}
    \widehat a(\theta)
    =\sum_{x\in\IZ^d}a_x e^{i x\cdot \theta},
    \qquad \theta=(\theta_1,\ldots,\theta_d)\in\IT^d\,.
    \label{eq:Fourier_transform_convention}
\end{align}

Then, $\Phi$ acts on $L^2(\IT^d,\IC^{2})$ by multiplication by the matrix-valued symbol
\begin{align}
    A(\theta)
    :=\Phi(e^{i\theta_1},\ldots,e^{i\theta_d})
    =\sum_{x\in \Gamma}\Phi_x e^{i x\cdot\theta}.
    \label{eq:Phi_symbol}
\end{align}

\begin{definition}\label{def:uniform_hyperbolicity}
We call $\Phi\in Sp(R,2)$ \emph{uniformly hyperbolic} if there exists $\delta>0$ such that the spectrum of each fiber $A(\theta)$ is uniformly separated from the unit circle:
\begin{align}
    \sigma(A(\theta))
    \cap
    \{\zeta\in\IC: e^{-\delta}<|\zeta|<e^{\delta}\}
    =\varnothing,
    \qquad \forall\theta\in\IT^d .
    \label{eq:uniform_hyperbolicity}
\end{align}

\end{definition}

The uniform hyperbolicity condition is a direct analog of the condition satisfied by the Anosov diffeomorphisms \cite{anosov1969geodesic}. Below we demonstrate that, similar to Anosov diffeomorphisms splitting the tangent bundle into stable/unstable subbundles, uniformly hyperbolic GQCAs split $P$ into stable/unstable subspaces. 

Assume from now on that $\Phi$ is uniformly hyperbolic.  Choose any radius $\rho$ with $e^{-\delta}<\rho<e^\delta$ where $\delta$ is taken from the preceding definition.  Since no eigenvalue of $A(\theta)$ crosses the circle $|\zeta|=\rho$, the Riesz projection
\begin{align}
    \Pi_-(\theta)
    :=\frac{1}{2\pi i}\oint_{|\zeta|=\rho}
      (\zeta I-A(\theta))^{-1}\,d\zeta
    \label{eq:Riesz_projection_stable}
\end{align}
exists for every $\theta\in\IT^d$ and is independent of the choice of $\rho$ inside the spectral gap.  It projects onto the generalized eigenspaces of $A(\theta)$ with $|\lambda|<1$.  We also set
\[
    \Pi_+(\theta):=I- \Pi_-(\theta),
\]
which projects onto the generalized eigenspaces with $|\lambda|>1$.  The maps $\theta\mapsto  \Pi_\pm(\theta)$ are continuous, and in fact real analytic on $\IT^d$. Indeed, the symbol $A(\theta)$
is real analytic on $\mathbb T^d$ and extends holomorphically to a
complex neighborhood of $\mathbb T^d$, since it is a Laurent polynomial
in $e^{i\theta_1},\ldots,e^{i\theta_d}$, see \cite{kato1995perturbation}, Chapter II.

Note that the projections commute with the symbol:
\begin{align}
     \Pi_\pm(\theta)A(\theta)=A(\theta) \Pi_\pm(\theta).
\end{align}
Therefore they define bounded projections $\Pi_{\pm}$ on $P_\IC$ by
\begin{align}
    (\Pi_\pm f)(\theta):= \Pi_\pm(\theta)f(\theta)\,.
    \label{eq:global_Riesz_projection}
\end{align}
Defining
\[
    P_{\IC,\pm}:=\operatorname{Ran}\Pi_\pm,
\]
we obtain the $A(\theta)$-invariant splitting
\begin{align}
    P_\IC=P_{\IC,-}\oplus P_{\IC,+}.
\end{align}
\begin{lemma}\label{lemma:exp_estimates}
Let $\Phi$ be uniformly hyperbolic and let $\delta$ be the constant from Definition \ref{def:uniform_hyperbolicity}. For every $0<\gamma<\delta$
there is $C_\gamma<\infty$ such that, for all $k\geqslant 0$,
\begin{align}
        \|\Phi^k f_-\|_2&\leqslant C_\gamma e^{-\gamma k}\|f_-\|_2,
        \qquad f_-\in P_{\mathbb C,-},\nonumber\\
        \|\Phi^{-k} f_+\|_2&\leqslant C_\gamma e^{-\gamma k}\|f_+\|_2,
        \qquad f_+\in P_{\mathbb C,+}.
        \label{eq:hyperbolic_estimates}
\end{align}
    
\end{lemma}
\begin{proof}
Choose $r_\gamma=e^{-\gamma}$. Since $\gamma<\delta$, the circle
$|z|=r_\gamma$ separates the stable spectrum of $A(\theta)$ from
the unstable spectrum for every $\theta\in\mathbb T^d$. By compactness,
\begin{align}
        M_\gamma
        :=
        \sup_{\theta\in\mathbb T^d,\ |z|=r_\gamma}
        \|(zI-A(\theta))^{-1}\|
        <\infty .
\end{align}
The holomorphic functional calculus gives
\begin{align}
        A(\theta)^k\Pi_-(\theta)
        =
        \frac{1}{2\pi i}\int_{|z|=r_\gamma}
        z^k(zI-A(\theta))^{-1}\,dz .
\end{align}
Hence
\begin{align}
        \sup_{\theta\in\mathbb T^d}
        \|A(\theta)^k\Pi_-(\theta)\|
        \leqslant C_\gamma e^{-\gamma k}.
\end{align}
Integrating the resulting pointwise bound over $\mathbb T^d$ gives
the first estimate. The second estimate follows by applying the same
argument to $A(\theta)^{-1}$ on $\operatorname{Ran}\Pi_+(\theta)$,
whose spectrum is also contained in $\{|z|\leqslant e^{-\delta}\}$.
    
\end{proof}

Equivalently, nonzero vectors in $P_{\IC,+}$ grow exponentially under positive powers of $\Phi$, while nonzero vectors in $P_{\IC,-}$ grow exponentially under negative powers. Also note that the estimates need not hold with $C_\gamma=1$ in the original $\ell^2$ norm.

Having obtained a projection acting on the complex Hilbert space $P_{\IC}$, we can pass to the real phase space.  Since the coefficients of $\Phi$ are real,
\begin{align}
    A(-\theta)=\overline{A(\theta)}.
\end{align}
Consequently,
\begin{align}
    \Pi_\pm(-\theta)=\overline{\Pi_\pm(\theta)}.
\end{align}
Thus $\Pi_\pm$ preserve the real subspace $P\subset P_\IC$, and we get an invariant splitting into real Hilbert subspaces
\begin{align}
    P=P_-\oplus P_+,
    \qquad
    P_\pm:=P\cap P_{\IC,\pm}.
    \label{eq:real_hyperbolic_splitting}
\end{align}

\begin{remark}
 The Riesz projection $\Pi_-$ is generally not a finite-range operator: its Fourier coefficients usually do not have finite support.  The uniform spectral gap implies exponential off-diagonal decay of those coefficients, so the projection is quasi-local rather than strictly local.
\end{remark}

Let
\begin{align}
    \Omega(a,b):=\sum_{n\in\IZ^d} a_n^{\mathsf T}Jb_n,
    \qquad a,b\in P,
    \label{eq:real_symplectic_form_l2}
\end{align}
be the real symplectic form, extended complex bilinearly to $P_\IC$.  In Fourier variables,
\begin{align}
    \Omega(f,g)
    =\int_{\IT^d} f(-\theta)^{\mathsf T}Jg(\theta)\,
      \frac{d\theta}{(2\pi)^d}\,.
    \label{eq:Fourier_symplectic_form}
\end{align}

\begin{proposition}\label{prop:intro_hyperbolic_splitting}
    Let $\Phi$ be a uniformly hyperbolic Gaussian QCA.  Then there is a $\Phi$-invariant splitting
    \begin{align}
        P=P_-\oplus P_+
    \end{align}
    into two Lagrangian subspaces.  Moreover, $\Phi$ is exponentially contracting on $P_-$, and $\Phi^{-1}$ is exponentially contracting on $P_+$.
\end{proposition}

\begin{lemma}\label{lem:hyperbolic_lagrangians}
The real subspaces $P_-$ and $P_+$ are Lagrangian subspaces of $(P,\Omega)$.
\end{lemma}

\begin{proof}
Equation \eqref{eq:fiber_symplectic_condition} implies that $\Phi$ preserves $\Omega$.  If $f,g\in P_-$, then
\[
    \Omega(f,g)=\Omega(\Phi^k f,\Phi^k g)
\]
for every $k\geq0$.  By \eqref{eq:hyperbolic_estimates}, the right-hand side tends to zero as $k\to\infty$, hence $\Omega(f,g)=0$.  Thus $P_-$ is isotropic.  The same argument applied to $\Phi^{-1}$ shows that $P_+$ is isotropic.  Since $P=P_-\oplus P_+$ and $\Omega$ is nondegenerate, each isotropic summand is maximal isotropic, hence Lagrangian.
\end{proof}

In summary, any uniformly hyperbolic GQCA splits $P$ into two lagrangians consisting of stable and unstable local (as well as non-local) Weyl labels. 

\paragraph{Local hyperbolicity}
\begin{definition}\label{def:local_hyperbolicity}
We call $\Phi\in Sp(R,2)$ \emph{locally hyperbolic} if there exists $\theta_0\in \IT^d$  such that $A(\theta_0)$ is hyperbolic, i.e., there exists $\delta>0$:
\begin{align}
    \sigma(A(\theta_0))
    \cap
    \{\zeta\in\IC: e^{-\delta}<|\zeta|<e^{\delta}\}
    =\varnothing\,.
    \label{eq:local_hyperbolicity}
\end{align}
\end{definition}

Since the hyperbolicity condition, $\det(zI-A(\theta_0))\neq 0$ for all $z\in S^1$, is open, and $A$ is continuous, there is an open neighborhood $U$ of $\theta_0$ such that $A(\theta)$ is hyperbolic for $\theta \in U$. Then, for each $\theta \in U$ we have Riesz spectral projection $\Pi_-(\theta)$ which continuously depends on $\theta$. We denote the corresponding projection on $U$ by $\Pi_{-}^U$. We obtain the same results about the splitting of $P$, now only for a range of frequencies. 

\subsection{Regularity}
In this section we formulate two conditions on GQCAs. Their role is to guarantee that any local Weyl label contains an unstable component such that the appropriate dynamics expands their norm.

 \begin{definition}
 Let $\a$ be a GQCA and $\Phi \in \operatorname{Sp}(R,2)$ the corresponding symplectic transformation. 
 \begin{enumerate}
  \item We say that a GQCA $\Phi$ is not regular if there exist $0\neq\xi \in V$, $c\in \IR^{\times}$, and $x\in \La$ such that 
   \begin{align}
       \a(W(\xi))=e^{i \chi(\xi)}\tau_{x}W(c\xi)\,.
   \end{align}
   Equivalently, 
   \begin{align}
       \Phi\xi=cT_x\xi\,.
   \end{align}
Otherwise, we call $\a$ (and $\Phi$) \textbf{regular}. 
     \item Assume that $\Phi$ is uniformly hyperbolic and let
$P=P_-\oplus P_+$ be the stable/unstable splitting defined by the
Riesz projection. We say that $\Phi$ is \textbf{everyday} if
\begin{align}
        V\cap P_-=\{0\}.
\end{align}
Equivalently, every nonzero strictly local label $\xi\in V$ has a
nonzero unstable component:
\begin{align}
        \Pi_+\xi\neq 0.
\end{align}
  
   \end{enumerate}
\end{definition} 
The regular GQCAs are discussed in \cite{G_tschow_2010} and in \cite{kapustin2026chaosthermalizationcliffordfloquetdynamics}. We are going to show that these two definitions are equivalent as long as we work with the on-site Hilbert space $\CH=L^2(\IR)$. 

We define the field of rational functions $ K:=\IC(u_1,\ldots,u_d)$ and notice that $A$, the symbol of $\Phi$, as a Laurent polynomial matrix, can be viewed as an endomorphism of $K^2$. 
\begin{proposition}\label{prop:fractality_regularity}
    The uniformly hyperbolic GQCA $\a$ (or $\Phi$) is everyday if and only if it is regular. 
\end{proposition}

\begin{proof}
Let $R_{\IC}:=R\otimes_{\IR}\IC\cong\IC[u_1^{\pm1},\ldots,u_d^{\pm1}]$
and identify $V_{\IC}:=V\otimes_{\IR}\IC$ with $R_{\IC}^2$.  Under this
identification $T_x$ is multiplication by the monomial $u^x$.  Since the Riesz
projection defining $P_-$ is real,
\begin{align}
        P_-\cap V=\{0\}
        \quad\Longleftrightarrow\quad
        P_{\IC,-}\cap V_{\IC}=\{0\}.
        \label{eq:regularity_complexification}
\end{align}
We prove the contrapositive in both directions.

First suppose that $\Phi$ is not everyday.  Choose
$0\neq\xi\in P_-\cap V$ and view $\xi$ as an element of $R_{\IC}^2$.  Since
$\CH=L^2(\IR)$, the space $\operatorname{Ran}\Pi_-(\theta)$ is
one-dimensional for every $\theta\in\IT^d$.  Hence, whenever
$\widehat\xi(\theta)\neq0$, the vectors $A(\theta)\widehat\xi(\theta)$ and
$\widehat\xi(\theta)$ are proportional. Thus, in $K^2$,
\begin{align}
        A(u)\hat \xi(u)=\mu(u)\hat \xi(u)\,,\quad \mu\in K\,.
    \end{align}
    
The eigenvalue $\mu(u)$ is a root of characteristic polynomial $\chi_A(z,u)=\det(zI-A(u))=z^2-z\Tr[A(u)]+\Det[A(u)]$. We observe that $\mu(u)\in R_{\IC}$: we substitute $\mu=\frac{f}{g}$ into the characteristic equation and find that $g|f^2$ but $\operatorname{gcd}(f,g)=1$ and $R_{\IC}$ is an UFD, which forces $g$ be a unit. Then, indeed, $\mu(u)$ is a Laurent polynomial. 

The second root of $\chi_A(z,u)$ is $\nu\coloneqq \Tr[A(u)]-\mu(u)$ and $\mu(u)\nu(u)=\Det[A(u)]\in R_{\IC}^{\times}$. Therefore, $\mu(u)$ is invertible and there exist $c\in \IC{^\times}$ and $x\in \La$ such that 
\begin{align}
    \mu(u)=cu^x\,.
\end{align}
 Since $A\hat \xi$ and
$\hat \xi$ have real coefficients, in fact $c\in\IR^\times$.  Consequently,
\begin{align}\label{eq:eigen_equation}
 A(u)\hat \xi(u)=cu^x\hat \xi(u)\,\quad \mbox{and}\quad       \Phi\xi=cT_x\xi,
\end{align}
with $0\neq\xi\in V$, so $\Phi$ is not regular.

Conversely suppose that $\Phi$ is not regular, so that
$0\neq\xi\in V$, $c\in\IR^\times$, and $x\in\La$ satisfy \eqref{eq:eigen_equation}.  Uniform hyperbolicity excludes
$|c|=1$.  If $|c|<1$, then $\xi$ belongs to the stable spectral subspace, so
$0\neq\xi\in P_-\cap V$ and $\Phi$ is not everyday.

It remains to consider $|c|>1$.  Put $\lambda(u)=cu^x$.  Since $A$ is symplectic,  $\det[A(u^{-1})]\det[ A(u)]=1$ and since $A$ is real, $|\det[ A(u)]|=1$ on $\IT^{d}$. Therefore,  the other eigenvalue $
        \nu(u):=\frac{\det A(u)}{\lambda(u)}$
has modulus $|c|^{-1}<1$ on $\IT^d$.  Choose a nonzero eigenvector
$w\in K^2$ for $\nu(u)$ and clear denominators to obtain
$0\neq\eta\in R_{\IC}^2$ with $A(u)\eta=\nu(u)\eta$.  Then
$\eta\in P_{\IC,-}\cap V_{\IC}$.  By
\eqref{eq:regularity_complexification}, $P_-\cap V\neq\{0\}$, so $\Phi$ is not
everyday.  We have shown that non-everydayness and non-regularity are equivalent,
which proves the proposition.
\end{proof}

\begin{remark}
    This proposition would fail if we considered the on-site Hilbert space $L^2(\IR^r)$ with $r>1$. Namely, any uniformly hyperbolic everyday GQCA would be regular but not any uniformly hyperbolic regular GQCA would be everyday. 
\end{remark}

\begin{example}[Regularity does not imply everydayness for two modes]
\label{ex:regular-not-everyday}
Let there be two bosonic modes per site, with canonical coordinates
$(q_1,q_2,p_1,p_2)$ and
\begin{align}
    J_4
    :=
    \begin{pmatrix}
        0&I_2\\
        -I_2&0
    \end{pmatrix}.
\end{align}
Set
\begin{align}
    Q
    :=
    \begin{pmatrix}
        0&-1\\
        1&0
    \end{pmatrix},
    \qquad
    S
    :=
    \begin{pmatrix}
        \frac12Q&0\\
        0&2Q
    \end{pmatrix}.
\end{align}
Then $S^{\mathsf T}J_4S=J_4$ and the constant-symbol, purely on-site GQCA
$A(\theta)=S$ is uniformly hyperbolic, with
\begin{align}
    \operatorname{spec}(S)
    =
    \left\{
        \frac{i}{2},
        -\frac{i}{2},
        2i,
        -2i
    \right\}.
\end{align}
Its stable subspace contains the real $q$-plane. Hence a nonzero
vector in this plane supported at a single site belongs to
$V\cap P_-$, so the GQCA is not everyday. It is nevertheless regular. Indeed, suppose that
$S\xi=cT_x\xi$ for a nonzero finitely supported real label
$\xi$ and $c\in\IR^\times$. Since $S$ acts invertibly and
independently at each site,
\begin{align}
    \supp(S\xi)=\supp(\xi).
\end{align}
A nonempty finite subset of $\IZ^d$ cannot be invariant under a
nonzero translation, so $x=0$. The equation would then give a real
eigenvector of $S$ with real eigenvalue $c$, but $S$ has no real
eigenvalues. Thus the GQCA is regular.

It is not thermalizing. If $0\neq\eta$ is a one-site label in the
stable $q$-plane, then
\begin{align}
    \Phi^k\eta
    =
    2^{-k}Q^k\eta
    \longrightarrow0.
\end{align}
For every locally normal state, strong continuity in that one-site
representation gives
\begin{align}
    \omega(W(\Phi^k\eta))
    \longrightarrow1.
\end{align}
Consequently, for any phase character,
\begin{align}
    |\omega(\alpha^k(W(\eta)))|
    \longrightarrow1,
\end{align}
so the dynamics is not thermalizing.
\end{example}

\begin{remark}[Purely on-site GQCAs are not everyday]
Let $\Phi$ be a translation-invariant Gaussian QCAs which is a tensor
product of identical on-site symplectic transformations. Equivalently,
its Laurent-polynomial symbol is constant:
\begin{align}
        A(\theta)=S,\qquad S\in \mathrm{Sp}(2,\mathbb R).
\end{align}
Assume that $\Phi$ is uniformly hyperbolic, so that $S$ has two eigenvalues $\l$ and $1/\l$ with $|\l|<1$. The stable Riesz projection for the
lattice system is then independent of $\theta$, and the global stable subspace is simply
\begin{align}
        P_-=\ell^2(\mathbb Z^d,E_-)
\end{align}
where $E_{-}$ is the eigenspace corresponding to eigenvalue $\l$. We have 
\begin{align}
    V \cap  P_{-} = c_{00}(\IZ^d,E_-)\neq \{0\}\,.
\end{align}

We find that a uniformly hyperbolic GQCAs which is merely a tensor product of
on-site symplectic transformations is never everyday.
\end{remark}

\section{Thermalization}\label{sec:thermalization}

In this section, we assemble the previous results into proofs of the main theorems. 
\begin{lemma}\label{lemma:uni_hyp_reg}
    Let $\Phi$ be uniformly hyperbolic and everyday. Then, $\lim_{k\to \infty}\|\Phi^k\xi\|_2=\infty$ exponentially fast for any nonzero $\xi\in V$.
\end{lemma}
\begin{proof}
    We decompose any $\xi\in V\subset P$ as $\xi=\Pi_{-}\xi+\Pi_{+}\xi$. Suppose that there exists $0\neq \xi\in V$ such that $\Pi_{+}\xi=0$. Then, we have that $V\ni \xi=\Pi_{-}\xi\in P_{-}$ and $P_-\cap V\neq 0$. By negating this, we obtain that if $
        P_{-}\cap V=0$, then for any $0\neq \xi \in V$, we have $\Pi_{+}\xi\neq 0$.

    Consider $\xi\in V$ such that $\Pi_{+}\xi\neq 0$. Since $P_+$ commutes with $\Phi$, we have
\begin{align}
\Phi^k \Pi_{+}\xi\in P_+
\qquad
\text{for all }k\geqslant 0.
\end{align}
By uniform hyperbolicity, for every $0<\gamma<\delta$ there exists
$C_\gamma<\infty$ such that
\begin{align}
\|\Phi^{k}\Pi_{+}\xi\|_2
\geqslant
C_\gamma^{-1} e^{\gamma k}\|\Pi_{+}\xi\|_2\,.
\end{align}

Since the Riesz projection $\Pi_+$ commutes with $\Phi$ and it is a bounded operator,
\begin{align}
\|\Phi^k \Pi_+ \xi\|_2
\leqslant
\|\Pi_+\|\,\|\Phi^k \xi\|_2.
\end{align}
Combining the two inequalities gives
\begin{align}
\|\Phi^k \xi\|_2
\geqslant
\frac{1}{\|\Pi_+\|}\|\Phi^k \Pi_+\xi\|_2
\geqslant
\frac{\|\Pi_+\xi\|_2}{C_\gamma\|\Pi_+\|}
e^{\gamma k}.
\end{align}
\end{proof}
We conclude that when the stable subspace $P_{-}$ of the Hilbert phase space $P$ is maximally non-algebraic, i.e., when $P_{-}$ contains no non-zero finitely supported labels, then $\Phi$ drives the $\ell^2$-norm of any local label to infinity. 

\begin{lemma}\label{lemma:_loc_hyp_fractality}
    Let $\Phi$ be locally hyperbolic and regular.  Then, $\lim_{k\to \infty}\|\Phi^k\xi\|_2=\infty$ exponentially fast for any nonzero $\xi\in V$.
\end{lemma}
\begin{proof}
Let $0\neq \xi\in V$.  Choose $\theta_0\in\IT^d$ such that
$A(\theta_0)$ is hyperbolic.  After shrinking the neighborhood $U$ of
$\theta_0$, we may assume that $A(\theta)$ is hyperbolic for all
$\theta\in U$, and we denote the corresponding stable and unstable Riesz
projections by $\Pi_-^U(\theta)$ and
$\Pi_+^U(\theta)=I-\Pi_-^U(\theta)$.

We claim that $\Pi_+^U\widehat\xi$ is not identically zero on $U$.
 Indeed, if
$\Pi_+^U\widehat\xi\equiv0$ on $U$, then
$\widehat\xi(\theta)\in\operatorname{Ran}\Pi_-^U(\theta)$ for
$\theta\in U$.  Since the fiber phase space is two-dimensional, this implies
that $A(\theta)\widehat\xi(\theta)$ and $\widehat\xi(\theta)$ are
proportional on $U$.  Hence
\begin{align}  \det\bigl(A(u)\widehat\xi(u),\widehat\xi(u)\bigr)
\end{align}
vanishes on a nonempty open subset of $\IT^d$, and therefore vanishes
identically as a Laurent polynomial.  Thus, over
$K$, we have
\begin{align}
        A(u)\widehat\xi(u)=\mu(u)\widehat\xi(u)
\end{align}
for some $\mu\in K$.  By the same algebraic argument as in the proof
of Proposition~\ref{prop:fractality_regularity}, $\mu$ is a unit of $R_{\IC}$. Hence $\Phi\xi=cT_x\xi$ for some $c\in\IR^\times, x\in\La$ which contradicts regularity. This proves the claim.

Choose a compact set $B\subset U$ of positive measure such that
\begin{align}
        g(\theta):=\Pi_+^U(\theta)\widehat\xi(\theta)
\end{align}
is nonzero in $L^2(B,\IC^2)$.  Since $B$ is compact and the fibers are
hyperbolic on $B$, the hyperbolicity over $B$ is uniform.  Therefore,
arguing as in Lemma~\ref{lemma:exp_estimates}, for every sufficiently small
$\gamma>0$ there exists $C_{B,\gamma}<\infty$ such that
\begin{align}
        \|A^kg\|_{L^2(B)}
        \geqslant C_{B,\gamma}^{-1}e^{\gamma k}\|g\|_{L^2(B)} .
\end{align}
Finally, $\Pi_+^U(\theta)$ commutes with $A(\theta)$, and
$\sup_{\theta\in B}\|\Pi_+^U(\theta)\|<\infty$.  Hence
\begin{align}
        \|\Phi^k\xi\|_2
        &=\|A^k\widehat\xi\|_{L^2(\IT^d)}
        \geqslant \|A^k\widehat\xi\|_{L^2(B)}  \geqslant
        \Bigl(\sup_{\theta\in B}\|\Pi_+^U(\theta)\|\Bigr)^{-1}
        \|\Pi_+^U A^k\widehat\xi\|_{L^2(B)}  \nonumber \\
        &=
        \Bigl(\sup_{\theta\in B}\|\Pi_+^U(\theta)\|\Bigr)^{-1}
        \|A^kg\|_{L^2(B)}  \geqslant c_\xi e^{\gamma k},
\end{align}
with $c_\xi>0$.  Thus $\|\Phi^k\xi\|_2\to\infty$ exponentially fast.
\end{proof}
\begin{proof}[Proof of Theorem \ref{th:main}]
Let $0\ne\xi\in V$, and set
\begin{align}
        \xi_k:=\Phi^k\xi,
        \qquad
        \Gamma_k:=\operatorname{supp}\xi_k.
 \end{align}
The scalar phase in
\begin{align}
        \alpha^k(W(\xi))=e^{i\chi_k(\xi)}W(\xi_k)
 \end{align}
has modulus one, so it suffices to prove
\begin{align}
        \omega(W(\xi_k))\to0.
 \end{align}

By everydayness, $\Pi_+\xi\ne0$. Choose $0<\gamma<\delta$.
Using the unstable estimate,
\begin{align}
        \|\Phi^k\Pi_+\xi\|_2
        \geqslant
        C_\gamma^{-1}e^{\gamma k}\|\Pi_+\xi\|_2.
 \end{align}
Since $\Pi_+\Phi^k\xi=\Phi^k\Pi_+\xi$,
\begin{align}
        \|\xi_k\|_2
        \geqslant
        \frac{\|\Pi_+\xi\|_2}{C_\gamma\|\Pi_+\|}
        e^{\gamma k}
        =:c_\xi e^{\gamma k}.
 \end{align}

Let $S\subset\mathbb Z^d$ be a finite spread set for $\Phi$. Then
\begin{align}
        \Gamma_k\subset \operatorname{supp}\xi+kS,
 \end{align}
so there is $C_\xi<\infty$ such that
\begin{align}
        |\Gamma_k|\leqslant C_\xi(1+k)^d.
 \end{align}
By Theorem \ref{thm:rl-intro} and Corollary \ref{cor: gqca bounds to zero},
\begin{align}
        |\omega(W(\xi_k))|
        \leqslant
        \frac{\pi\sqrt{2\nu|\Gamma_k|+1}}{\|\xi_k\|_2}
        \leqslant
        \frac{\pi\sqrt{2\nu C_\xi(1+k)^d+1}}
             {c_\xi e^{\gamma k}}
        \longrightarrow 0.
 \end{align}
This proves the theorem.    
\end{proof}
\begin{proof}[Proof of Theorem \ref{th:main_local}]
The proof is the same after we use
Lemma \ref{lemma:_loc_hyp_fractality}.
\end{proof}

\subsection{Family of Examples}\label{sec:example}

Let us consider a family of simple one-dimensional examples. Let $d=1$ and $R=\mathbb R[u,u^{-1}]$ and define a family of GQCAs by the following symplectic transformation
\begin{align}
    \Phi_a(u):=
    \begin{pmatrix}
        0 & -1\\
        +1 & t_a
    \end{pmatrix}\,,\quad t_a=a+u+u^{-1}\,.
 \end{align}
We check that $\Phi_a(u^{-1})^T J\Phi_a(u)=J$, 
so $\Phi_a\in \operatorname{Sp}(R,2)$. 
The eigenvalues of $A(\theta)$ are the roots of the characteristic polynomial $\chi_a(\theta,z)=z^2-z(a+2\cos\theta)+1$. 
The fiber
$A_a(\theta)$ is hyperbolic exactly when 
\begin{align}
    |t_a(\theta)|>2.
    \label{eq:example_fiber_hyp_condition}
\end{align}
Then, 
\begin{align}
    \Phi_a \text{ is uniformly hyperbolic}
    \quad\Longleftrightarrow\quad
    |a|>4.
    \label{eq:example_uniform_hyp_locus}
\end{align}
On the other hand,
\begin{align}
    \Phi_a \text{ is locally hyperbolic}
    \quad\Longleftrightarrow\quad
    a\neq0.
    \label{eq:example_local_hyp_locus}
\end{align}

We now check regularity. The characteristic polynomial is reducible over $K=\IC(u)$ if and only if its discriminant is a
square in $K$.  The discriminant is
\begin{align}
    \Delta_a(u)
    &:=t_a(u)^2-4  \\
    &=u^{-2}\bigl(u^2+(a-2)u+1\bigr)
             \bigl(u^2+(a+2)u+1\bigr).
    \label{eq:example_discriminant}
\end{align}
The factor $u^{-2}$ is a square in $K$.  The two quadratic factors in
\eqref{eq:example_discriminant} are coprime, because their difference is
$4u$ and neither factor vanishes at $u=0$.  Hence their product is a square
if and only if both quadratic factors are squares. Further, one checks
\begin{align}
    \chi_a(z,u) \text{ is reducible over } K
    \quad\Longleftrightarrow\quad
    a=0.
    \label{eq:example_reducibility_locus}
\end{align}
For $a=0$ one has
\begin{align}
    \chi_0(z,u)=(z-u)(z-u^{-1})
\end{align}
and, for instance,
\begin{align}
    \Phi_0(u)\binom{1}{-u}
    =u\binom{1}{-u}\,,
\end{align}
 so $\Phi_0$ is not regular.  Therefore,
\begin{align}
    \Phi_a \text{ is regular}
    \quad\Longleftrightarrow\quad
    a\neq0.
    \label{eq:example_regularity_locus}
\end{align}

By Proposition \ref{prop:fractality_regularity}, a uniformly hyperbolic $\Phi_a$ is everyday iff it is regular.  We conclude that $\Phi_a$ is everyday for every $|a|>4$. In summary, $\Phi_a$ is locally hyperbolic and regular for any $a\neq0$ and uniformly hyperbolic and everyday for any $a$ such that $|a|>4$.

\section{Discussion}

We have proved thermalization for a class of translation-invariant Gaussian quantum cellular automata acting on bosonic lattice systems in infinite volume. The convergence is formulated on the algebraic local Weyl algebra: for every locally normal initial state with uniformly bounded particle density, the expectation of every nontrivial local Weyl operator converges
to zero under the dynamics.

We carried out the analysis of quantum many-body dynamical systems by transferring the problem to the classical phase space. In \cite{kapustin2026chaosthermalizationcliffordfloquetdynamics} this approach was dubbed QCA/CA (cellular automata) correspondence. In general, even classically, results about the dynamics on non-compact phase spaces are scarce. In particular, the methods used in loc. cit. would not apply to our case. Nevertheless, the statement of our thermalization Theorem \ref{th:main_local} is similar to Theorem 2.2 of \cite{kapustin2026chaosthermalizationcliffordfloquetdynamics}, and, in the classical case, to Theorem 3.1 of \cite{kapustin2026thermalizationmanybodyclassicalfloquet}. 

We distinguish the local Weyl algebra $\CA_{\mathrm{loc}}$ considered in this work from
the larger quasilocal algebra of bounded operators at each site
\begin{align}
    \mathfrak B_{\mathrm{ql}}
    :=
    \overline{
      \bigcup_{\Gamma\Subset\La}
      \CB(\CH_\Gamma)
    }^{\|\cdot\|}.
\end{align}
The limiting functional $\omega_\infty$ is tracial on the algebraic
local Weyl algebra $\CA_{\mathrm{loc}}$. By contrast, $\mathfrak B_{\mathrm{ql}}$ admits no
tracial state. Thus our result cannot be
interpreted as convergence to an infinite-temperature trace on
$\mathfrak B_{\mathrm{ql}}$. The estimates proved here concern Weyl
observables only; they neither produce an extension of
$\omega_\infty$ to $\mathfrak B_{\mathrm{ql}}$ nor make a claim about
convergence of arbitrary quasilocal observables.

Several natural questions remain. First, the everydayness condition is algebraic and checkable in examples, but a broader characterization of its prevalence in higher-dimensional parameter families of Laurent-polynomial symplectic maps would be desirable. Second, the proof yields no sharp thermalization rate beyond the bound obtained from the hyperbolicity gap and support growth; a tighter estimate would require more detailed information on the fiber dynamics and the initial state's local number distribution. Third, it would be interesting to determine whether everydayness and hyperbolicity are stable under perturbations in the space of dynamics. Finally, the case when $\Phi$ has eigenvalues on the unit circle presents a natural avenue toward understanding localization: in the absence of hyperbolicity, the space of Weyl operators may fragment in ways that obstruct thermalization, potentially giving rise to localized or partially ergodic behavior.

\section*{Acknowledgements}
The authors would like to thank Anton Kapustin for helpful conversations and advice on the subject. Additionally, the authors thank the participants of the Mathematical Physics seminar at UCLA, in particular Sameer Erramilli, Niccol\`o Porciani, Rahul Roy, and Ryan Thorngren, for useful comments. The authors attribute insights on the proof of the quantum many-body Riemann-Lebesgue lemma, especially Proposition \ref{prop: quadratic forms bound}, to GPT-5.4 and GPT-5.5. The authors acknowledge support from the Mani L. Bhaumik Institute for Theoretical Physics.

\bibliography{InvStates}
\end{document}